\documentclass[aps,prx,twocolumn,showpacs,preprintnumbers,longbibliography]{revtex4-1}

\usepackage{graphicx}
\usepackage{amsfonts,amssymb,amsmath,amsthm}
\usepackage{hyperref}

\usepackage{wasysym} 

\newcommand{\CC}{\mathcal{C}}
\newcommand{\PP}{\mathcal{P}}
\newcommand{\cli}{\CC_d^{\otimes n}}
\newcommand{\pa}{\PP_d^{\otimes n}}
\newcommand{\ZZ}{\mathbb{Z}}
\newcommand{\nn}{\nonumber}
\newcommand{\omg}{\omega}
\newcommand{\mdag}{^\dagger}
\newcommand{\half}{\frac{1}{2}}
\newcommand{\hc}{\text{H.c.}}
\newcommand{\ra}{\rightarrow}
\newcommand{\px}{\sigma^x}
\newcommand{\py}{\sigma^y}
\newcommand{\pz}{\sigma^z}
\newcommand{\id}{\openone}
\newcommand{\LL}{\mathcal{L}}
\newcommand{\vod}{V_{\text{od}}}
\newcommand{\vd}{V_{\text{d}}}
\newcommand{\hvd}{\hat{V}_{\text{d}}}
\newcommand{\heff}{H_{\text{eff}}}

\newcommand{\ket}[1]{\left | #1 \right\rangle}
\newcommand{\bra}[1]{\left \langle #1 \right |}

\newtheorem{theorem}{Theorem}
\newtheorem{lemma}{Lemma}

\begin{document}

\title{Parafermions in a Kagome lattice of qubits for topological quantum computation}
\author{Adrian Hutter}
\author{James R. Wootton}
\author{Daniel Loss}
\affiliation{Department of Physics, University of Basel, Klingelbergstrasse 82, CH-4056 Basel, Switzerland}

\date{\today}

\begin{abstract}
Engineering complex non-Abelian anyon models with simple physical systems is crucial for topological quantum computation. Unfortunately, the simplest systems are typically restricted to Majorana zero modes (Ising anyons). Here we go beyond this barrier, showing that the $\mathbb{Z}_4$ parafermion model of non-Abelian anyons can be realized on a qubit lattice. 
Our system additionally contains the Abelian $D(\mathbb{Z}_4)$ anyons as low-energetic excitations. We show that braiding of these parafermions with each other and with the $D(\mathbb{Z}_4)$ anyons allows the entire $d=4$ Clifford group to be generated. The error correction problem for our model is also studied in detail, guaranteeing fault-tolerance of the topological operations. Crucially, since the non-Abelian anyons are engineered through defect lines rather than as excitations, non-Abelian error correction is not required. Instead the error correction problem is performed on the underlying Abelian model, allowing high noise thresholds to be realized.
\end{abstract}

\maketitle

\section{Introduction}

Non-Abelian anyons exhibit exotic physics that would make them an ideal basis for topological quantum computation \cite{kitaev,nayak,pachos_book}. 
It has recently become apparent that truly scalable quantum computation with non-Abelian anyons can only be achieved when invoking active error correction, despite the protection provided by a finite anyon gap \cite{woottonNA,brell,pedrocchi}.
The development of practical systems in which non-Abelian anyons may be created, manipulated, and detected is therefore highly important.
Systems in which non-Abelian anyons arise typically suffer from one of two drawbacks: either they are experimentally extremely challenging to realize (as is the case for quantum double \cite{kitaev} or string-net models \cite{levin}), or it is not clear how they can be made compatible with the active error correction required for fault-tolerance (as is the case for FQH systems).

A particularly attractive approach for building a fault-tolerant quantum computer is to use a system of physical qubits (spin-$\half$ particles). A number of technologies allow for precise qubit control, such as superconducting qubits \cite{barends}, trapped atomic ions \cite{monroe}, spin qubits \cite{kloeffel}, or cold atoms or polar molecules in optical lattices \cite{negretti}. A qubit lattice with two-body nearest neighbour interactions would therefore be an ideal system to realize non-Abelian anyons. 
While the model we propose involves three-body interactions, we discuss how these can be obtained from two-body interactions.

Non-Abelian anyons supported by a qubit system typically are Majorana zero modes, also known as Ising anyons \cite{kitaev_honey,yao,bombin,petrova}.
A variety of proposals for experimental realization of Majorana zero modes in solid state systems have also been developed \cite{alicea}.
These anyons can be used to perform universal quantum computation when enhanced by non-topological operations \cite{bravyi,freedman}. However, these additional operations are highly resource intensive. Anyon models with a richer set of topological operations would therefore be much more practical for the realization of topological quantum computation.
Here we solve this by introducing a model composed of two-qubit Hamiltonian interactions that can realize a more complex model of non-Abelian anyons, known as $\ZZ_4$ parafermions. The error correction problem for these is studied in detail.

Parafermion modes are generalizations of Majoranas whose fusion and braiding behavior is more complex and computationally more powerful. This has led to a quest in recent years for systems that could host them. Numerous proposals for their experimental implementation in condensed matter systems such as fractional quantum Hall systems, nanowires, or topological insulators have recently appeared \cite{lindner,cheng,clarke,vaezi,burrello,mong,barkeshli1,klinovaja1,zhang,oreg,klinovaja2,klinovaja3,barkeshli2,orth,klinovaja4,alicea_para}.

Extrinsic defects in \emph{Abelian} topological states can behave like non-Abelian anyons \cite{barkeshli_theory}. 
The idea of non-Abelian anyons at the ends of defect lines, first introduced for FQH states \cite{barkeshli3}, has been adapted to the $D(\ZZ_d)$ quantum double models in Refs.~\cite{you_first,you}.
These anyons are Majorana zero-modes for $d=2$ and more powerful parafermions for $d>2$.
Unfortunately, the generalized Pauli operators appearing in the $D(\ZZ_d)$ quantum double models models coincide with the physically relevant spin-operators only for $d=2$. 
Otherwise, their structure makes them highly difficult to realize experimentally. The case $d=4$, however, allows us to combine the best of both worlds.
The joint Hilbert space of two qubits allows the $4$-dimensional generalized Pauli operators to be expressed in terms of two-qubit operators. Using this, we show how $\ZZ_4$ parafermions can emerge in a lattice of qubits with nearest-neighbor interactions only. This allows the computational power of $\ZZ_4$ parafermions to be harnessed in a qubit system.

The fact that our system is built on top of a system supporting Abelian anyons (the $D(\ZZ_4)$ quantum double model) proves very useful.
The non-Abelian parafermion modes can not only be braided with each other, but also with Abelian excitations of the quantum double model, allowing us to generate the entire Clifford group for $d=4$ by quasi-particle braiding. 
This extends beyond the limited set of gates found using the same parafermions in previous work \cite{clarke}.
Furthermore, we do not have to perform non-Abelian error correction (a still poorly understood problem \cite{woottonNA,brell,hutter,wootton_hutter,burton,hutter_proof}) to guarantee fault-tolerance, but can correct the underlying Abelian model.
This Abelian error correction problem is nevertheless more involved than the well-studied error correction problem for the standard $D(\ZZ_d)$ models, and we study it in detail.

The rest of this paper is organized as follows. 
In Sec.~\ref{sec:operators} we show how $\ZZ_4$ parafermion operators can be expressed in terms of qubit operators.
Sec.~\ref{sec:model} introduces a qubit Hamiltonian whose low-energetic excitations correspond to the $D(\ZZ_4)$ quantum double model.
In Sec.~\ref{sec:defect} we discuss how $\ZZ_4$ parafermion modes appear at the ends of defect strings in our model.
We demonstrate in Sec.~\ref{sec:nonab} how the non-Abelian braiding statistics of these modes can be used to perform logical gates.
Appendix~\ref{app:clifford} contains a proof that the set of gates which can be performed this way generates the entire Clifford group $\mathcal{C}_4$, which may be of independent interest.
In Sec.~\ref{sec:errcorr} we study the error correction problem of our model in detail and conclude in Sec.~\ref{sec:conclusions}.

\section{$\ZZ_4$ parafermion operators in terms of qubit operators}\label{sec:operators}

We consider $d$-dimensional generalizations of the Pauli matrices $X$ and $Z$. These are unitary operators satisfying $X^d = Z^d = \id$ and $ZX = \omega XZ$, where $\omega=e^{2\pi i/d}$ with integer $d>1$. If we define $Y=\omg^{(d+1)/2}X\mdag Z\mdag$, we also have $Y^d=\id$, $XY=\omg YX$, and $YZ=\omg ZY$. Operators $X_i$ and $Z_i$ act on qudit $i$ and hence $[X_i,X_j]=[Z_i,Z_j]=[X_i,Z_j]=0$ if $i\neq j$.

These operators are related to those of parafermions. Given a total ordering on the qudits $\lbrace i\rbrace$, one can obtain parafermion operators via a non-local transformation \cite{fradkin}
\begin{align}\label{eq:trafo}
 \gamma_{2i-1} = (\prod_{j<i}X_j)Z_i\,, \quad \gamma_{2i} = \omega^{(d+1)/2}(\prod_{j\leq i}X_j)Z_i\,.
\end{align}
These satisfy the $\ZZ_d$ parafermion relations,
\begin{align}\label{eq:parafermion}
 \gamma_j^d=\id\,, \quad \gamma_j\gamma_k = \omega^{\text{sgn}(k-j)}\gamma_k\gamma_j\,.
\end{align}

The operators $X$, $Y$, and $Z$ can be represented as $d$-dimensional matrices. It is thus natural to seek a representation of these operators for the case $d=4$ on the Hilbert space of two qubits (spins-$\half$).
Indeed, given two qubits $1$ and $2$, one easily verifies that the operators
\begin{align}\label{eq:XYZ}
 X &= \half(\px_1+\px_2-i\pz_1\py_2+i\py_1\pz_2) \nn\\
 Y &= \half e^{i3\pi/4}(\py_1+i\py_2+i\px_1\pz_2+\pz_1\px_2) \nn\\
 Z &= \frac{1}{\sqrt{2}}e^{i\pi/4}(\pz_1-i\pz_2) 
\end{align}
are $4$-dimensional generalized Pauli operators, and $\ZZ_4$ parafermions can be obtained from these via Eq.~(\ref{eq:trafo}).
We also note that $X^2 = \px_1\px_2$, $Y^2 =\py_1\py_2$, and $Z^2 = \pz_1\pz_2$.

\begin{figure}
\centering
\includegraphics[width=1.00\columnwidth]{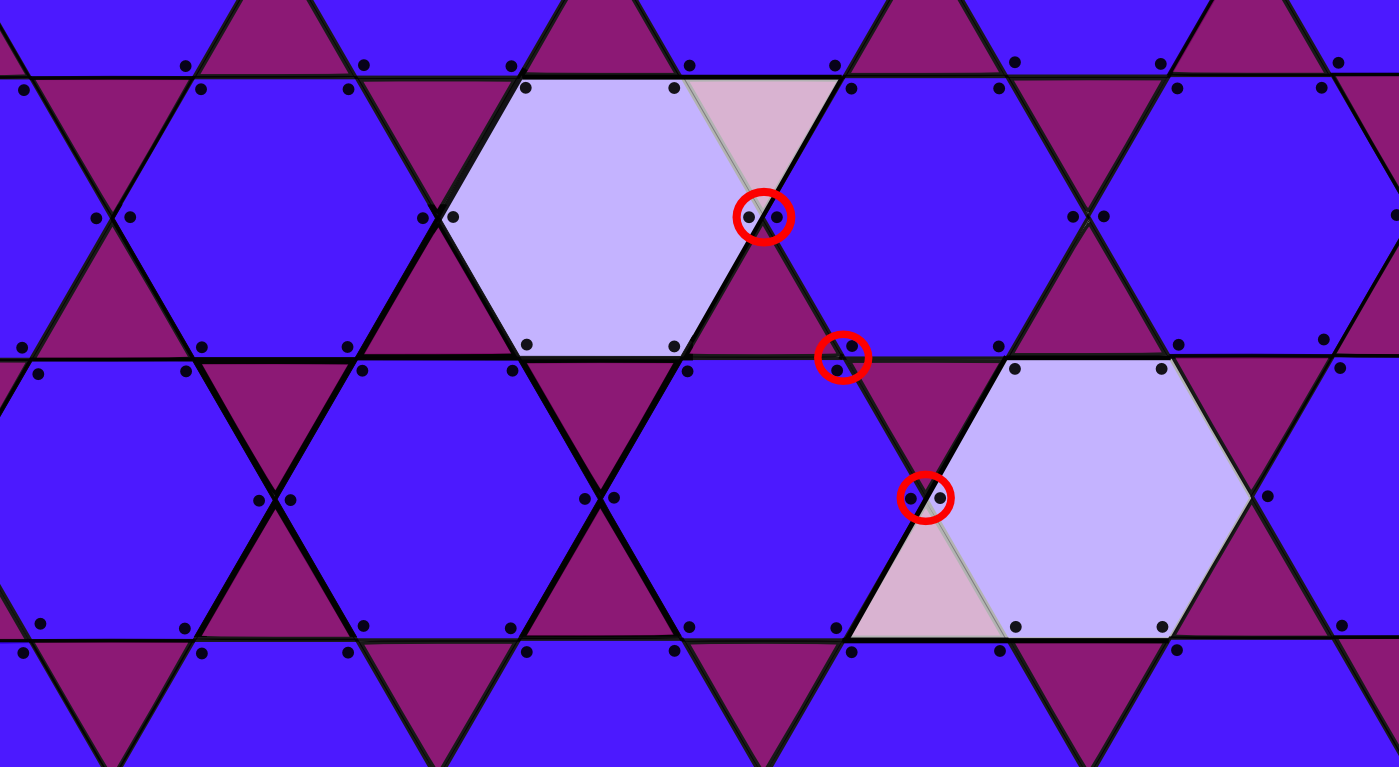}
\caption{Two qubits are located at each vertex of a Kagome lattice. Each pair of qubits hosts two $\ZZ_4$ parafermions. To unlock their potential for non-Abelian braiding, two such parafermions need to become unpaired, which is achieved by adding a defect line to the lattice. These are strings of strong local operators acting on qubit pairs (encircled). They create unpaired parafermion modes located at their ends (light pentagon-shaped regions consisting of a hexagon and a triangle).
}
\label{fig:kagome}
\end{figure}

\section{Model}\label{sec:model}

We consider a two-dimensional Kagome (trihexagonal) lattice as in Fig.~\ref{fig:kagome}. Each vertex of the lattice hosts one $4$-dimensional qudit (one pair of $\ZZ_4$ parafermions) or, in other words, two qubits.
The Hamiltonian of our model is given by
\begin{align}\label{eq:ham}
 H = \sum_{\triangle}H_{\triangle} + h\sum_i(\px_{i1}+\px_{i2})\,.
\end{align}
Here, the first term is a sum of equivalent terms for each triangle in the Kagome lattice. 
We label the vertices around one triangle $a$, $b$, and $c$, and the two qubits which are present at vertex $a$ are called $a_1$ and $a_2$, etc. The triangle terms in the Hamiltonian are then given by
\begin{align}
 H_{\triangle} &= \frac{J}{2}(\pz_{a1}\pz_{b1}\pz_{c1}+\pz_{a2}\pz_{b2}\pz_{c2}) \nn\\&\quad -\frac{J}{2}(\pz_{a1}\pz_{b1}\pz_{c2}+\pz_{a1}\pz_{b2}\pz_{c1}+\pz_{a2}\pz_{b1}\pz_{c1}  \nn\\&\quad +\pz_{a2}\pz_{b2}\pz_{c1}+\pz_{a2}\pz_{b1}\pz_{c2}+\pz_{a1}\pz_{b2}\pz_{c2}) \,.
\end{align}
The second sum $\sum_i$ in Eq.~(\ref{eq:ham}) runs over all vertices in the lattice. The two qubits located at vertex $i$ are called $i1$ and $i2$.
This second sum thus represents a uniform magnetic field in $x$-direction. 

Our Hamiltonian involves three-qubit terms of the form $\pz_a\pz_b\pz_c$. It is in principle straightforward to generate these from one-body terms and two-body interactions by use of perturbative gadgets \cite{kempe,jordan,oliveira}.
Consider a ``mediator qubit'' $u$ coupled to qubits $a$, $b$, and $c$.
Starting from a Hamiltonian
\begin{align}
 H_{\text{gadget}}=-\frac{\Delta}{2}\pz_u+\alpha(\pz_a+\pz_b)\px_u+\beta\pz_c\pz_u+\gamma\pz_a\pz_b+\delta\pz_c\,,
\end{align}
and consider the perturbative regime $\Delta\gg|\alpha|, |\beta|$. In this regime, it is possible to integrate out qubit $u$. Taking up to third-order terms into account, one finds an effective Hamiltonian
\begin{align}
 H_{\text{eff}} = (\beta+\delta)\pz_c + (-2\frac{\alpha^2}{\Delta}+\gamma)\pz_a\pz_b -4\frac{\alpha^2\beta}{\Delta^2}\pz_a\pz_b\pz_c\,.
\end{align}
Choosing $\delta=-\beta$ and $\gamma=2\frac{\alpha^2}{\Delta}$ produces the desired three-qubit term without any undesired one- or two-qubit terms.

The generation of three-body interactions in optical lattices has been discussed in detail in Refs.~\cite{pachos,buechler}. These proposals would make the perturbative gadgets unnecessary.
A ``toolbox'' for generating spin-lattice models such as ours in optical lattices has also been developed \cite{micheli}.
Generating non-Abelian anyons other than Majorana zero modes by use of perturbative gadgets from two-body interactions has previously been discussed in Refs.~\cite{kapit,brell_NJP}.

The spin-Hamiltonian in Eq.~(\ref{eq:ham}) can be exactly rewritten as
\begin{align}\label{eq:paraham}
 H = -J\sum_\triangle(Z_aZ_bZ_c + \hc) + h\sum_i(X_i+X_i\mdag)\,.
\end{align}
Here again the first sum runs over all triangles in the lattice and the corners of a triangle are labeled $a$, $b$, and $c$.
The second sum runs again over all vertices of the lattice.

We now consider the perturbative limit $h\ll J$ and regard the second sum in Eq.~(\ref{eq:paraham}) as a perturbation to the first term. 
Note that all terms in the first sum in Eq.~(\ref{eq:paraham}) commute, so the unperturbed Hamiltonian is trivially solved.
The lowest-order non-vanishing terms appear in sixth-order perturbation theory. We find an effective Hamiltonian
\begin{align}\label{eq:heff}
 \heff &= -J\sum_\triangle(Z_aZ_bZ_c + \hc) \nn\\&\quad -\frac{63}{8}\frac{h^6}{(2J)^5}\sum_{\hexagon}(X_rX\mdag_sX_tX\mdag_uX_vX\mdag_w+\hc)\,,
\end{align}
where the second sum runs over all hexagons in the Kagome lattice and $r$, $s$, $t$, $u$, $v$, $w$ label the six vertices around each hexagon.
The effective Hamiltonian in Eq.~(\ref{eq:heff}) is derived in Appendix~\ref{app:perturbation}.

We note that all summands in $\heff$ commute, so the system is exactly solvable. The excitations of this system are Abelian anyons corresponding to the $D(\ZZ_4)$ quantum double model.
The topological degeneracy of the model can be made manifest by studying non-local loop degrees of freedom that commute with all stabilizers $Z_aZ_bZ_c$, $X_rX\mdag_sX_tX\mdag_uX_vX\mdag_w$, and their Hermitian conjugates, and fullfil themselves $
\ZZ_4$ relations. A possible choice of operators is illustrated in Fig.~\ref{fig:operators}.

\begin{figure}[htb]
\centering
  \begin{tabular}{|c|c|}
    \hline
    \includegraphics[width=.45\columnwidth]{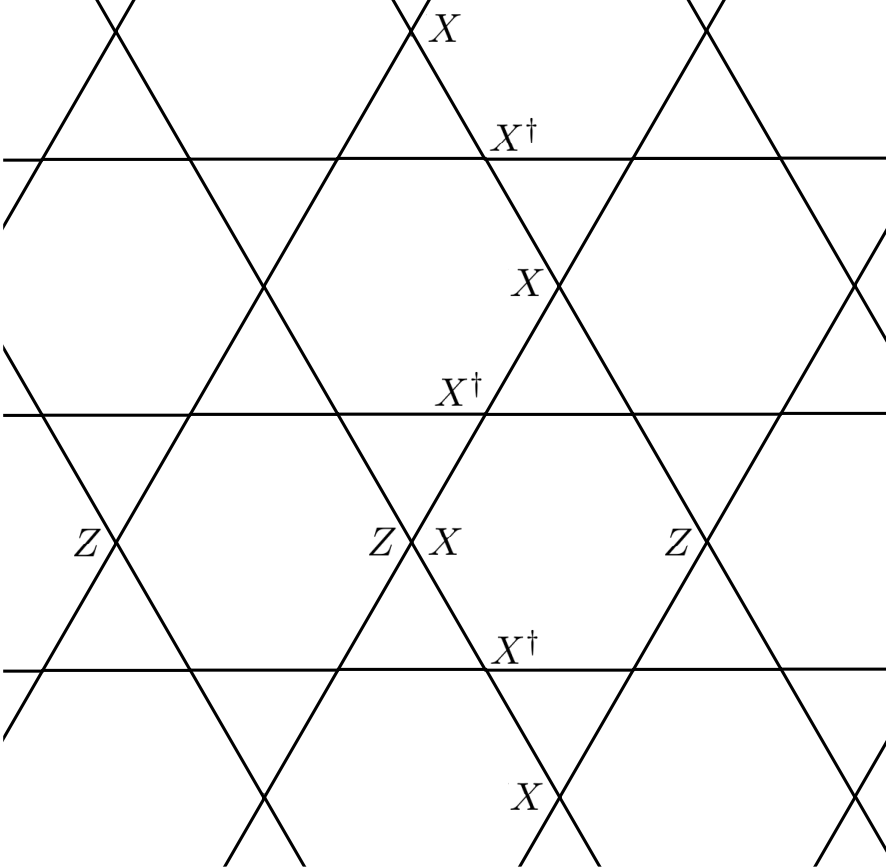} &
    \includegraphics[width=.45\columnwidth]{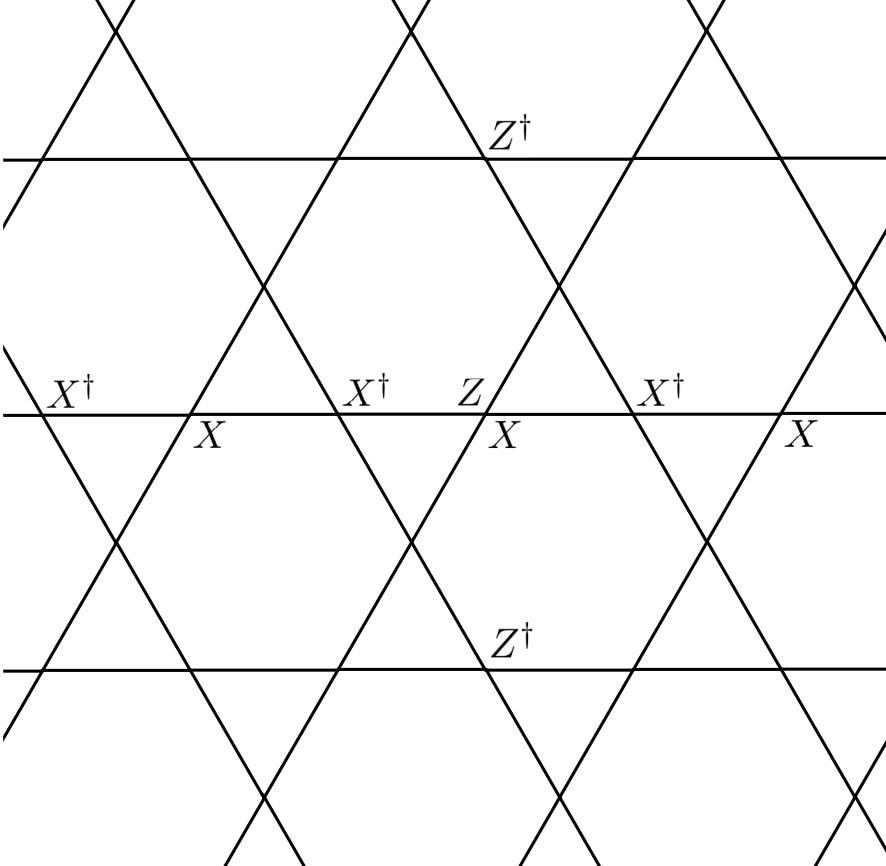}   \\
    \hline
  \end{tabular}
  \caption{Two sets of logical operators $\tilde{X}_1=XX\mdag XX\mdag\ldots$, $\tilde{Z}_1=ZZZZ\ldots$ (left figure) and $\tilde{X}_2=XX\mdag XX\mdag\ldots$, $\tilde{Z}_2=ZZ\mdag ZZ\mdag\ldots$ (right figure) that satisfy the commutation relations of $4$-dimensional generalized Pauli operators.}
  \label{fig:operators}
\end{figure}

In passing, we note that the $\ZZ_2$ version of Eq.~(\ref{eq:paraham}), in which the $\ZZ_4$ operators $X$ and $Z$ are replaced by Pauli operators $\px$ and $\pz$, leads to an effective Hamiltonian analogous to Eq.~(\ref{eq:heff}) and thus provides a very simple model with topological order.
While this model requires three-body operators $\pz\pz\pz$ as opposed to Kitaev's honeycomb Hamiltonian \cite{kitaev_honey} which involves two-body interactions only, all of these interactions connect the same spin-component, which may provide a significant practical simplification over the honeycomb model.

\section{Parafermion modes and defect lines}\label{sec:defect}

The model is constructed from the cyclic qudit operators $Z$ and $X$, which are related to parafermion operators. It is therefore natural to seek an interpretation of the model in terms of parafermionic modes.

To do this we must first fix the exact form of the stabilizers, which define the anyonic charge carried by each excitation. Let us use $E_p$ ($M_p$) to denote the stabilizer for a hexagonal (triangular) plaquette, $p$. For hexagonal plaquettes we use the convention that $E_p = X_rX\mdag_sX_tX\mdag_uX_vX\mdag_w$, where $r$ refers to the top-right corner and the other corners are labelled in an anti-clockwise fashion. For triangular plaquettes we use $M_p = Z_aZ_bZ_c$ for all triangles of the form $\triangle$ and $M_p = Z_a\mdag Z_b\mdag Z_c\mdag$ for all triangles of the form $\bigtriangledown$.
The stabilizer operators $E_p$ and $M_p$ are unitary operators with eigenvalues $\omg^k$, $k\in\lbrace0,1,2,3\rbrace$, where here and in the following $\omg=i$ for $d=4$.
An eigenvalue $\omega^g$ of the $E_P$ corresponds to a charge anyon of the form $e_g$, while $M_P$ similarly detects flux anyons $m_h$. Fusion of charge anyons forms a representation of $\ZZ_4$, as does that of fluxes. 
The convention for the stabilizer operators chosen before ensures that the anyonic charge of both charge and flux type anyons is independently conserved (modulo $4$).
A full clockwise monodromy of an $e_g$ around an $m_h$, or \emph{vice versa}, yields a phase $\omega^{gh}$, see Fig.~\ref{fig:phases} for illustration. 

\begin{figure}[htb]
\centering
  \begin{tabular}{c}
    \includegraphics[width=0.9\columnwidth]{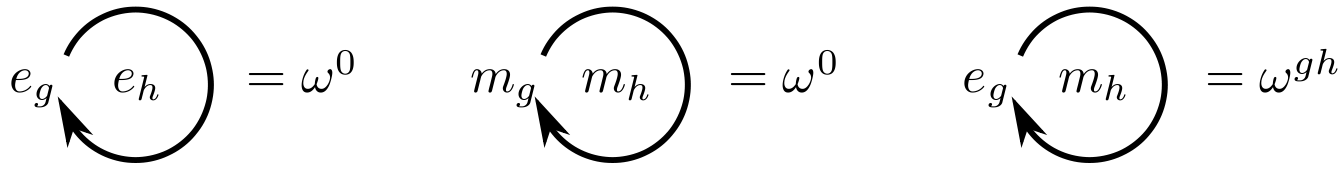} \\
    \quad \\
    \includegraphics[width=0.9\columnwidth]{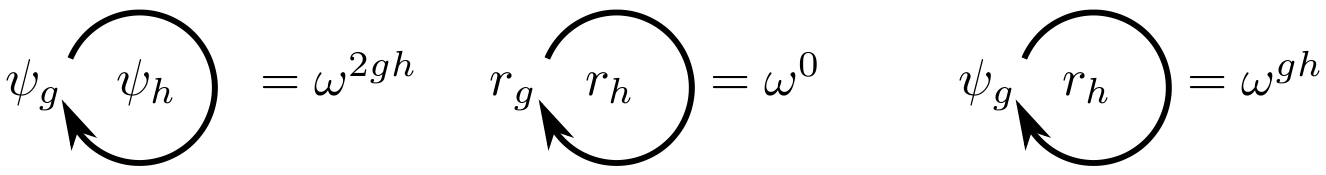} 
  \end{tabular}
  \caption{Phases obtained by braiding the $e$- and $m$-excitations of the $D(\ZZ_4)$ model around each other (top), and by braiding the excitations $\psi$ and $r$ of the transformed stabilizers around each other (bottom).}
  \label{fig:phases}
\end{figure}

Just as Majorana modes (Ising anyons) in the qubit toric code \cite{bombin,you}, parafermions appear in our system at the ends of defect strings.
For the interpretation in terms of parafermions, it will be useful to introduce a new set of composite anyons defined as $\psi_g = e_g \times m_g$. 
These also obey $\ZZ_4$ fusion with each other, and their braiding behavior can be inferred from the behavior of the constituent charge and flux particles.
The particles $\lbrace\psi_0, \psi_1, \psi_2, \psi_3\rbrace$ form a chiral Abelian anyon model with Chern number $\nu=2$ \cite{kitaev_honey}.

Note that 
\begin{align}\label{eq:composite}
 e_g\times m_h = \psi_g \times m_{h-g}\,.
\end{align}
We now perform a local transformation from the set of stabilizer generators $\lbrace E_p, M_p\rbrace$, detecting the charges on the left-hand-side of Eq.~(\ref{eq:composite}), 
to a new set $\lbrace S_p, R_p\rbrace$ which detects the two charges on the right-hand-side.
Let $H$ denote the set of hexagonal plaquettes and $T$ denote the set of triangular plaquettes. Note that $|T|=2|H|$.
Consider an injective map $\varphi:H\ra T$, which to each hexagonal operator $E_p$ assigns one of the six adjacent triangular operators $M_{\varphi(p)}$.
Typically, we choose $M_{\varphi(p)}$ to be the top-right neighbor of $E_p$, while other choices become necessary next to defect lines.
The transformation from the old to the new set of stabilizers reads $S_p=E_p$ for $p\in H$ and
\begin{align}
 R_p = \begin{cases}
	  M_pE_{\varphi^{-1}(p)}\mdag&\text{if }p\in\text{Im}(\varphi) \\
	  M_p&\text{if }p\notin\text{Im}(\varphi)
       \end{cases}
\end{align}
for $p\in T$. Here, $\text{Im}(\varphi)$ denotes the image of the map $\varphi$.

Since $\prod_{p\in H}S_p = \prod_{p\in T}R_p = \id$, the charges detected by the new stabilizers are separately conserved (modulo $4$).
Just like the $\psi_g$ anyons detected by the $S_p$ stabilizers, the $R_g$ charges detected by the $R_p$ stabilizers also form an anyon model obeying $\ZZ_4$ fusion.
However, while these two anyon models have the same fusion rules, they are not equivalent, as they exhibit different braiding behavior.
A full clockwise monodromy of a $\psi_g$ around a $\psi_h$ gives a phase of $\omg^{2gh}$, a monodromy of a $r_g$ around an $r_h$ gives a phase of $1$, and a monodromy of a $\psi_g$ around a $r_h$ gives a phase of $\omg^{gh}$, see again Fig.~\ref{fig:phases}.
Just like the $e_g$ and $m_h$ charges, the $\psi_g$ and $r_h$ particles correspond to a way of decomposing the $D(\ZZ_4)$ model into two submodels which are closed under fusion, but have non-trivial mutual braiding behavior,
\begin{align}
 D(\ZZ_4) &= \lbrace e_0, e_1, e_2, e_3\rbrace \times \lbrace m_0, m_1, m_2, m_3\rbrace \nn\\ 
&= \lbrace\psi_0, \psi_1, \psi_2, \psi_3\rbrace \times \lbrace r_0, r_1, r_2, r_3\rbrace\,,
\end{align}
where the three particle models other than $\lbrace\psi_0, \psi_1, \psi_2, \psi_3\rbrace$ correspond to the simple $\ZZ_4$ model.

The stabilizer operators $S_p$ detect the presence of $\psi_g$ anyon which are pinned to a pentagon-shaped double plaquette, made up of a neighbouring pair of triangular and hexagonal plaquettes. 
These anyons can be regarded as generalizations of Dirac fermions to the group $\ZZ_4$ (rather than $\ZZ_2$).
Just as Dirac modes can be decomposed into two Majorana modes, so too can the $\psi$ modes be decomposed into two parafermion modes. 
Two parafermion modes, $P_a$ and $P_b$, are therefore associated with each double plaquette, $P$. These are described using parafermion operators satisfying Eq.~(\ref{eq:parafermion}).
The parity operator for the $\psi$ mode associated with a pair $(j,k)$ is defined $ \omega^{(d+1)/2} \gamma_j \gamma\mdag_k$ for $j<k$, and so $S_P = \omega^{5/2} \gamma_{P_a} \gamma\mdag_{P_b}$.

For a stabilizer state, the system is within a definite eigenstate of all $S_P$. The parafermion modes are therefore all paired, with the pairs corresponding to the two within each double plaquette. 
In order to use the parafermion modes as non-Abelian anyons, some must be allowed to become unpaired. The creation and transport of unpaired parafermion modes can be done by adapting the method of Ref.~\cite{you} to the Kagome lattice. 
The method can be interpreted in terms of anyonic state teleportation \cite{bonderson,bonderson_interactions}, as explained for the Majorana case in Ref.~\cite{wootton}.

The method introduces unpaired parafermion modes at the endpoints of defect lines. These are lines on which additional single qudit terms are added to the Hamiltonian, of one of the two following forms
\begin{align}\label{eq:defects}
Y+\hc &= -\frac{1}{\sqrt{2}}(\py_1+\py_2+\px_1\pz_2+\pz_1\px_2)\,, \nn\\
\omega^{5/2}X Z\mdag+\hc &= \frac{1}{\sqrt{2}}(\py_1-\py_2+\px_1\pz_2-\pz_1\px_2)\,.
\end{align}
Specific examples are shown in Fig.~\ref{fig:defects}. 

The single qudit terms added along defect lines are much stronger than any other interactions, and thus effectively remove the qudits on which they act from the code.
This means that the $E_P$ and $M_P$ operators for the double plaquettes along these lines no longer commute with the Hamiltonian, and so can no longer be used as stabilizer generators. 
Their pentagon-shaped product, $R_P$, is used instead.
The pentagons in Fig.~\ref{fig:defects} show how next to a defect line the mapping $\varphi$ needs to pick the bottom-left triangular-shaped stabilizer of a hexagon-shaped stabilizer to ensure that their product still commutes with the Hamiltonian.

\begin{figure}
\centering
  \begin{tabular}{c}
  \quad \\
  \includegraphics[width=0.90\columnwidth]{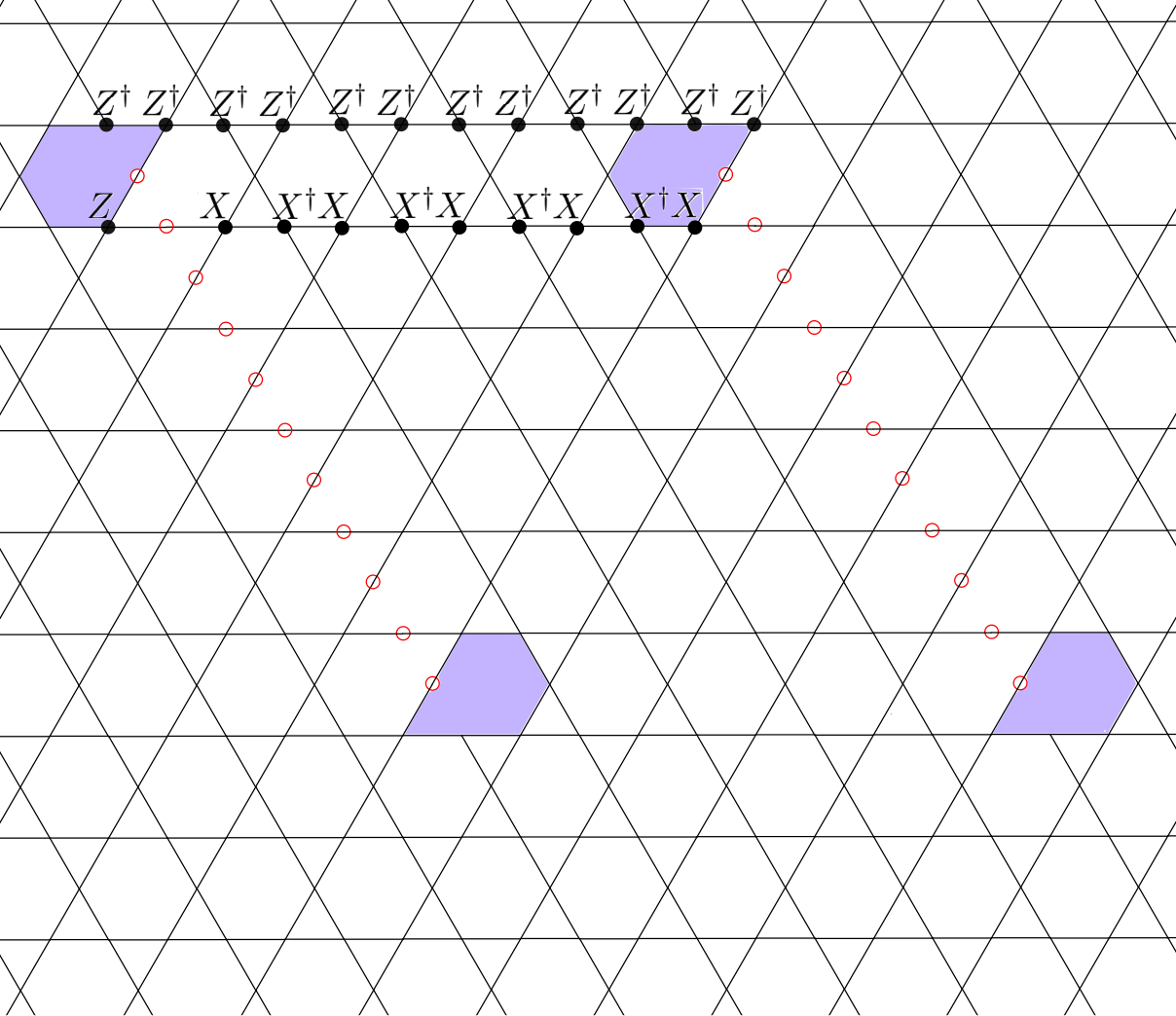} \\
  \quad \\
\includegraphics[width=0.90\columnwidth]{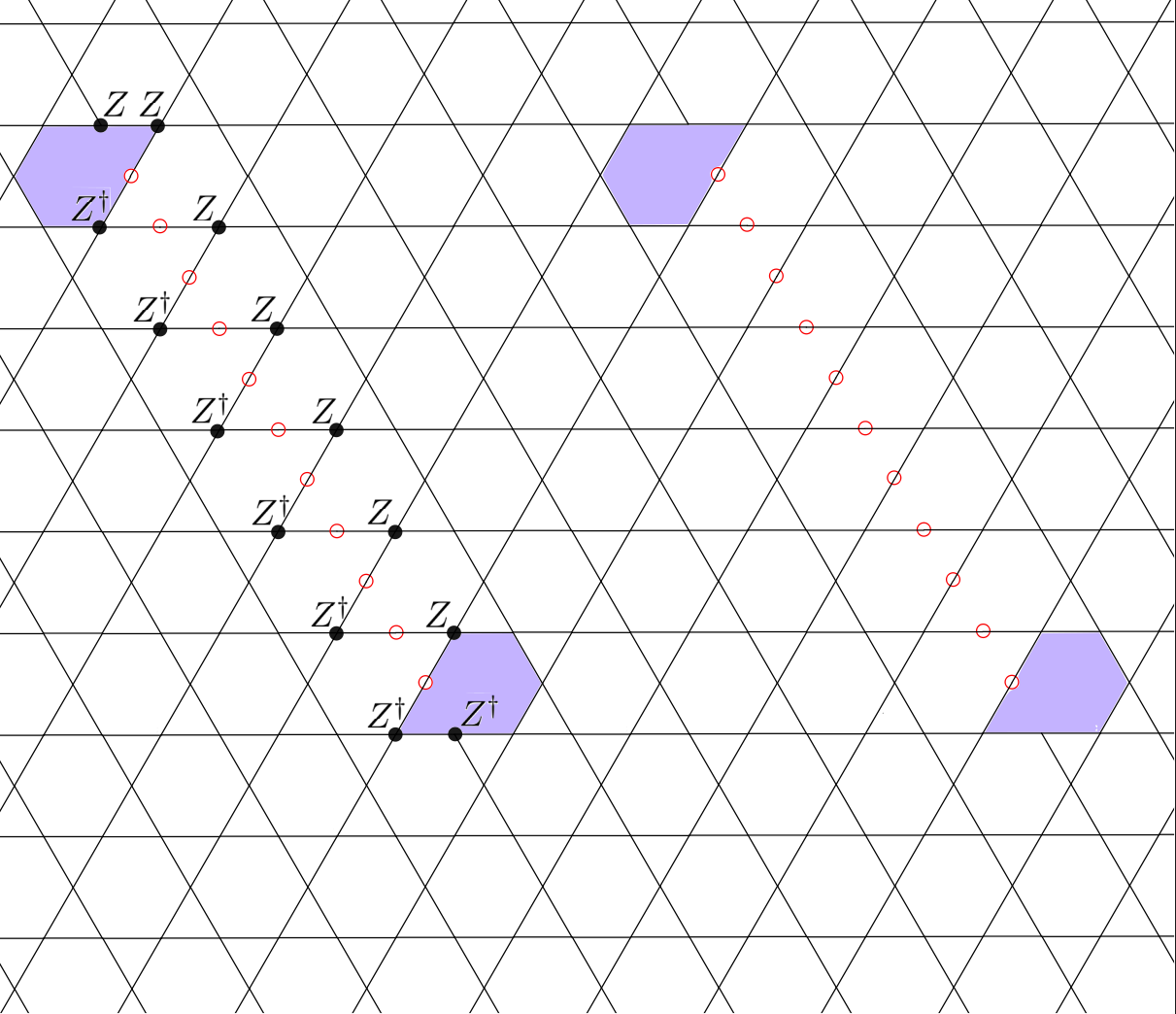} 
  \end{tabular}
  \caption{
  Strings of alternating single-qudit operators of the form $\omega^{5/2}Y_i+\hc$ or $\omega^{5/2}X_i Z_i\mdag+\hc$ (encircled) are added to the Hamiltonian. 
  These effectively eliminate the qudits on which they act from the code, leading to enlarged, pentagon-shaped stabilizers along the defect string.
  Parafermion modes reside on the pentagons at the ends of the defect strings (shaded).
  A possible choice for two logical operators $\tilde{X}_L$ (top) and $\tilde{Z}_L$ (bottom) satisfying $\tilde{Z}_L\tilde{X}_L = \omega\tilde{X}_L\tilde{Z}_L$ is illustrated.}
  \label{fig:defects}
\end{figure}

This change of the stabilizer generators of the code has a drastic effect. Consider an $e_g$ anyon moved towards a point along a defect line from one direction, and an $m_g$ moved towards the same point from the other direction. 
Both of these are detected by $R_P$ type stabilizers.
When they meet on the same double plaquette, they will fuse to form a $\psi_g$, and so not be detected by the $R_P$ stabilizers anymore. In fact, since the $S_P$ stabilizer is removed for double plaquettes along a defect line, they will not be detected by any stabilizer operator. The $\psi_g$ occupancy of a defect line corresponds to an increased groundstate degeneracy of the system, referred to as a \emph{synthetic topological degeneracy} \cite{you}.

In the following section, the $\lbrace\psi_g, r_h\rbrace$ decomposition of the $D(\ZZ_4)$ model will prove more convenient than the $\lbrace e_g, m_h\rbrace$ decomposition.
A process in which a defect line converts an $m_g$ into an $e_{-g}$ can equivalently be described as one in which a $r_g$ passes a defect line which emits a $\psi_{-g}$.

\section{Parafermions as non-Abelian anyons}\label{sec:nonab}

Since unpaired parafermion modes reside at the endpoints of defect strings, it is natural to use them to explain the properties of the modified stabilizer. Parafermion modes are described by a non-Abelian anyon model with particle species $\lbrace\psi_0, \psi_1, \psi_2, \psi_3, \sigma\rbrace$. Here $\psi_0\equiv1$ corresponds to the anyonic vacuum and $\sigma$ is an unpaired parafermion mode. The fusion rules of this anyon model are
\begin{align}\label{eq:fusion_rules}
 \sigma \times \sigma &= \psi_0 + \psi_1 + \psi_2 + \psi_3\,,\nn\\
 \psi_g \times \psi_h &= \psi_{g\oplus h}\,,\nn\\
 \psi_g \times \sigma &= \sigma\,,
\end{align}
where $\oplus$ denotes addition modulo $4$.
A pair of parafermions (or the defect line between them) may therefore collectively hold any of the four types of $\psi$ anyon. 

The anyon model with the fusion rules given in Eq.~(\ref{eq:fusion_rules}) does not allow for a non-trivial solution of the pentagon and hexagon equations.
As such, it obeys only projective non-Abelian statistics. 
The computational power of braiding parafermions has recently been studied in Ref.~\cite{hutter_para}.
In our setup, we cannot only braid the parafermions with each other, but can also braid the Abelian $e$- and $m$-particles around them, which provides the possibility to perform additional gates.
In the following, we want to study the gate set that can be generated this way.

As in the Majorana/Ising case, we use four parafermion modes (two defect strings) for which the total fusion sector is vacuum to store one logical qudit. The natural logical operators are parity operators for the pairs of parafermions. An eigenvalue $\omega^g$ corresponds to a $\psi_g$ occupancy for the pair, and so the specific result $\sigma \times \sigma  = \psi_g$ if they would be fused. We associate the $Z$ basis of the logical qudit with the $\psi$ occupancy of vertical pairs (connected by defect lines). 

Specific choices of logical operator are illustrated in Fig.~\ref{fig:defects}. The $\tilde{Z}_L$ corresponds to a clockwise loop of an $e_1$ around a defect line. The braiding of this $e_1$ around the $\psi_g$ held in the pair yields the required phase of $\omega^g$. The $\tilde{X}_L$ corresponds to clockwise loop of an $e_{-1}$ anyon which is converted to an $m_1$ through one defect line and back to an $e_{-1}$ through the other.
Equivalently, we can describe it as a clockwise loop of a $r_1$ and a transfer of a $\psi_1$ from the right to the left defect line. 

Let us denote a state in which the left defect line holds a mode $\psi_g$ and the right defect line holds a mode $\psi_h$ by $\ket{\psi_g,\psi_h}$.
Two defect lines create a $4\times4$-fold synthetic topological degeneracy. 
For computational purposes, we restrict to the $4$-dimensional subspace of states of the form $\ket{g}_L\equiv\ket{\psi_g,\psi_{-g}}$.
This is the set of states which can locally be created from the anyonic vacuum.
The effect of the logical operators on these states is $\tilde{X}_L\ket{g}_L=\ket{g\oplus1}_L$ and $\tilde{Z}_L\ket{g}_L=\omg^g\ket{g}_L$.

In addition to the logical operators $\tilde{X}_L$ and $\tilde{Z}_L$, which can be performed in our model by braiding the \emph{Abelian} $D(\ZZ_4)$ anyons around the parafermion modes (ends of defect strings), 
we can perform further topologically protected single-qudit and two-qudit gates by braiding the parafermion modes themselves. 
Defect lines used for braiding are shown in Appendix~\ref{app:braiding}. 
Crucially, braiding parafermions allows one to perform an entangling gate by topological means, which is in contrast to Majorana fermions \cite{clarke}. 
What is more, exploiting the fact that our non-Abelian system is built on top of an Abelian $D(\ZZ_4)$ system allows us to generate the entire $4$-level Clifford group by braiding quasi-particles, as we discuss in the following.

For the rest of this section, $X$ and $Z$ refer to the logical operators called $\tilde{X}_L$ and $\tilde{Z}_L$ before, respectively.
The first column in Fig.~\ref{fig:braiding} illustrates how braiding of $D(\ZZ_4)$ charges and fluxes can be used to perform logical $X$ and $Z$ gates.
Whether an $e_1$ or an $m_1$ anyon is used to perform the logical $Z$ gate is irrelevant.

\begin{figure}
 \includegraphics[width=.8\columnwidth]{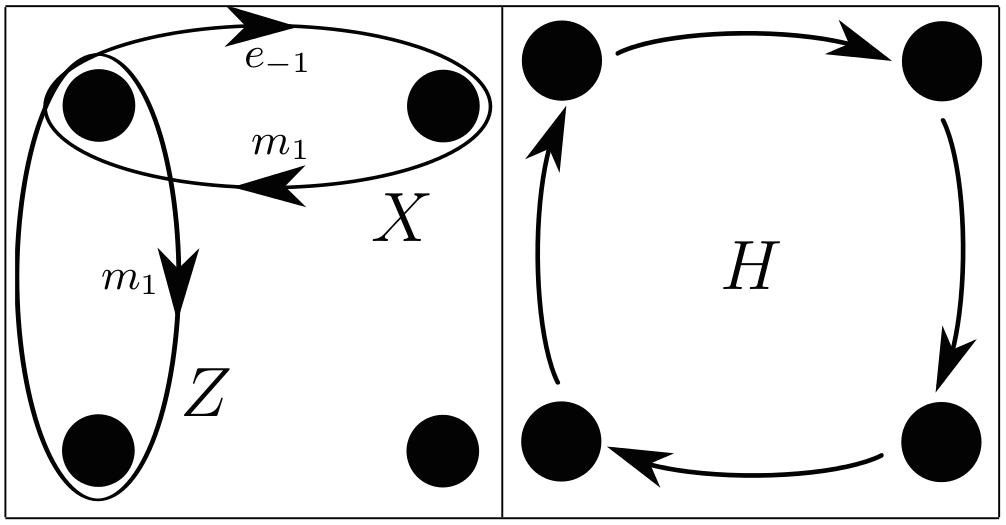}
  \caption{All generators of the single-qudit Clifford group can be performed by braiding quasi-particles. The four circles correspond to the four parafermion modes which are used to store one logical qudit. 
  The left part of the figure illustrates how to perform the logical operators $X$ and $Z$ by braiding the Abelian excitations of the $D(\ZZ_4)$ model around the parafermions. 
  The right part demonstrates a logical Hadamard gate $H$, which is performed by braiding the parafermion modes themselves.}
  \label{fig:braiding}
\end{figure}

Consider two ?parafermion modes storing a $\psi_g$ particle. A full clockwise monodromy of one parafermion around the other can be understood as a monodromy of the constituent $e_g$ around the $m_g$, yielding an $\omg^{g^2}$ phase.
We can thus expect a single exchange of the two parafermion modes storing a $\psi_g$ to yield a square root of this phase, such as $\omg^{g^2/2}$. This is demonstrated directly by studying the necessary microscopic operations in App.~\ref{app:exchange}.

For a logical qudit stored in four parafermion modes, let $S$ denote a clockwise exchange of a vertical pair of parafermion modes, and $T$ an exchange of a horizontal pair, see Fig.~\ref{fig:generators}. 
As discussed, we have $S=\sum_g\omg^{g^2/2}\ket{g}\bra{g}$, while $T$ is diagonal in the logical $X$ basis. In the logical $Z$ basis, $T$ reads (for $d=4$)
\begin{align}
 T = \half e^{-i\pi/4}\sum_{gh}e^{i\frac{\pi}{4}(g-h)^2}\ket{g}\bra{h}\,.
\end{align}

\begin{figure}[htb]
\centering
  \begin{tabular}{|c|c|}
    \hline
    \includegraphics[width=0.35\columnwidth]{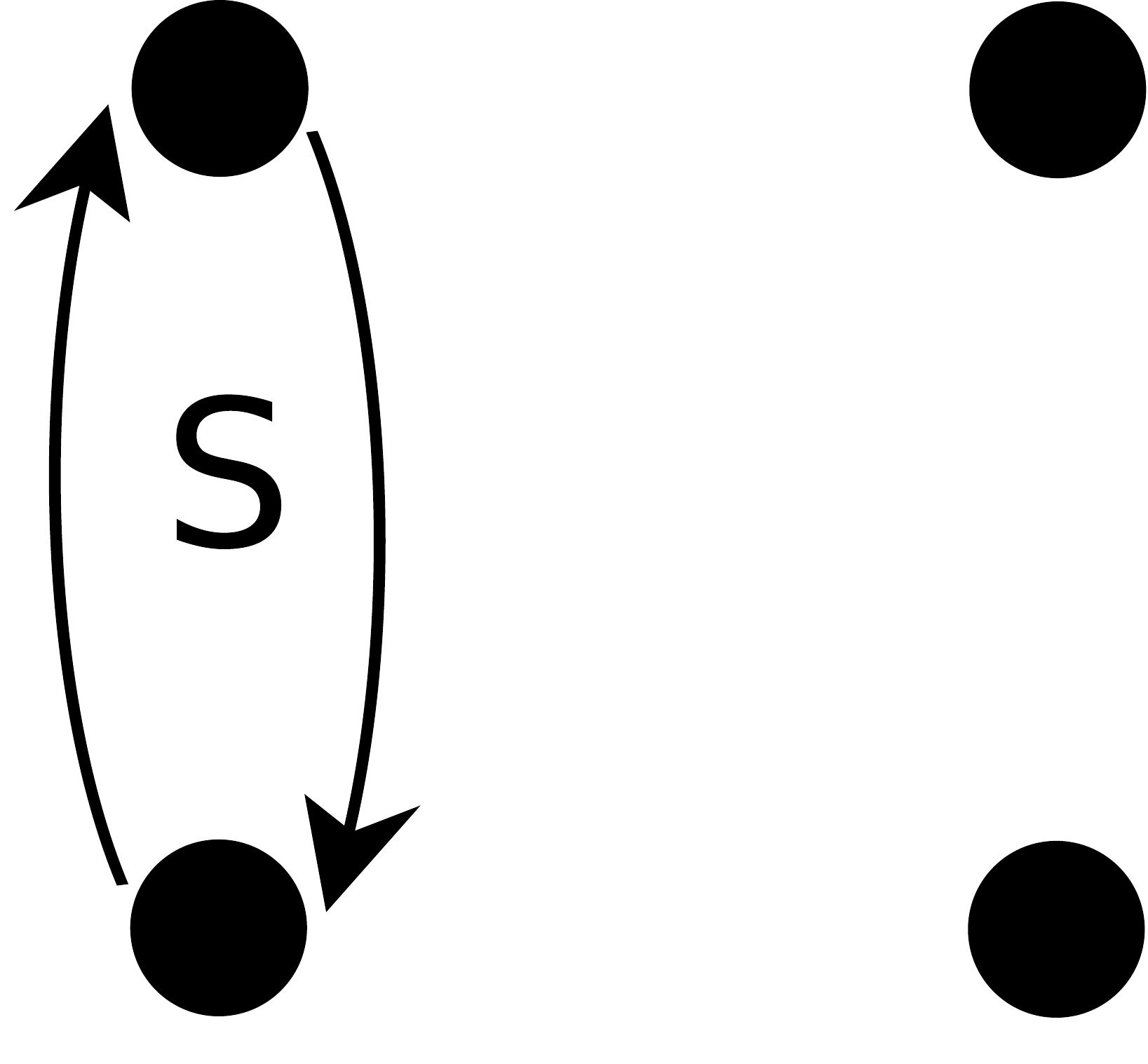} &
    \includegraphics[width=0.31\columnwidth]{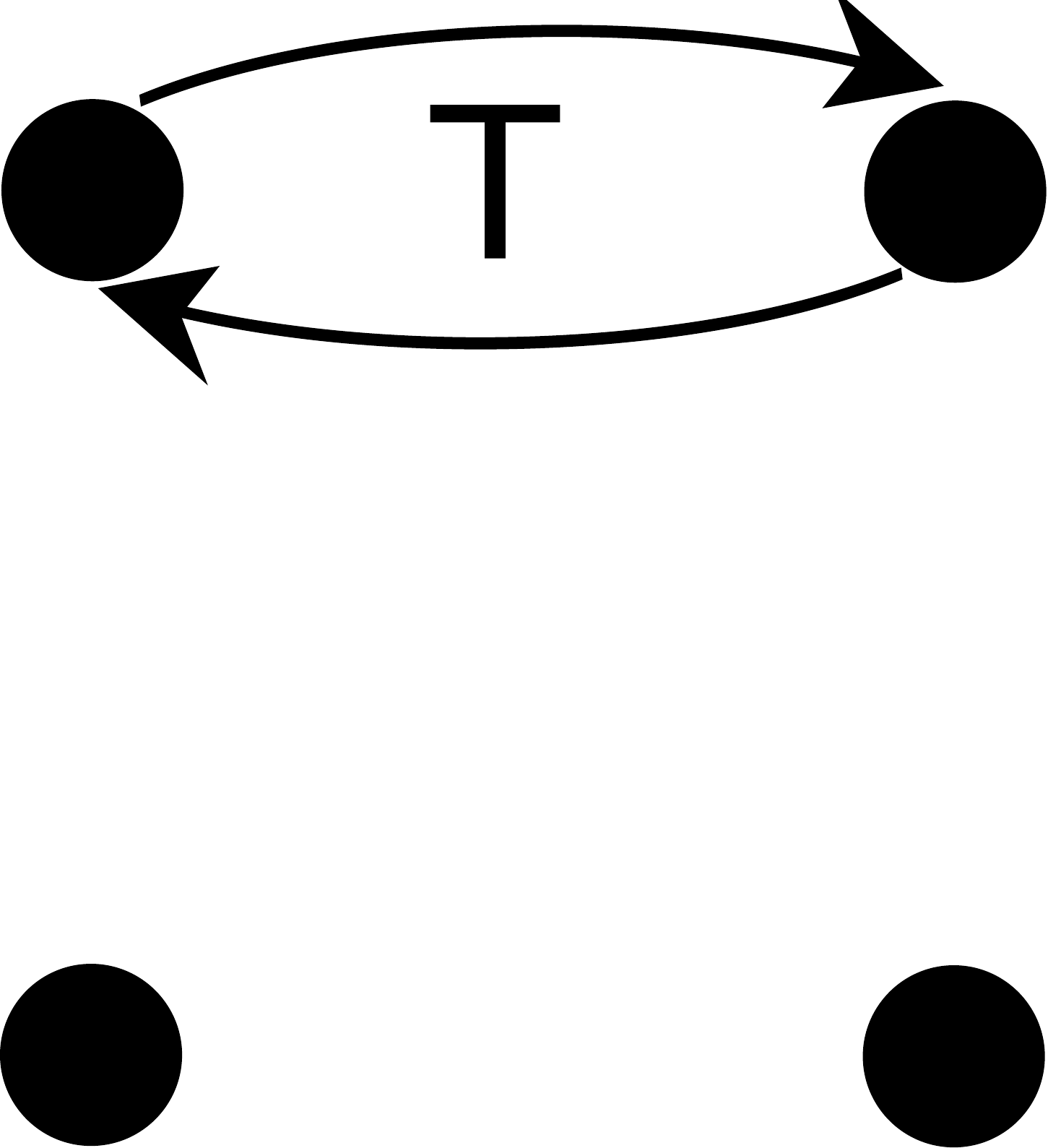} \\
    \hline
  \end{tabular}
  \caption{Generators $S$ and $T$ of all gates that can be performed on a qudit stored in the fusion space of four parafermions by braiding them.}
  \label{fig:generators}
\end{figure}

Again, in contrast to Majorana fermions, parafermions support an entangling gate between two logical qudits by braiding operations \cite{clarke}. 
The controlled phase-gate $\Lambda$ is defined by its action on a logical two-qudit basis state, $\Lambda\ket{g,h}=\omg^{gh}\ket{g,h}$.
In our parafermion scheme, an entangling gate can be performed by braiding of a pair of parafermions from one qudit with a pair from the other. 
Let us consider, for example, the braiding of the left vertical pair for both qudits. 
For an initial logical product state $\ket{g,h}$, the process corresponds to braiding a $\psi_g$ clockwise around a $\psi_h$, which yields a phase of $\omega^{2gh}$. 
The resulting operation is therefore the squared controlled phase-gate $\Lambda^2$.
For $d=2$, corresponding to the Ising/Majorana case, $\Lambda^2=\id$, and so this operation is trivial. 
For $d>2$, however, it is a non-trivial entangling gate, akin to the one proposed in Ref.~\cite{clarke}.

Clearly a more powerful entangling gate would be $\Lambda$ itself. This can be achieved for $\ZZ_d$ parafermions for odd $d$ by taking the $(d+1)/2$-th power of $\Lambda^2$. 
However these do not admit the simple decomposition into qubits that we have used in defining the model. 
Fortunately, we can make use of the underlying charge and flux anyons to realize $\Lambda$ despite the even qudit dimension.

The defect line may be interpreted as a hole for $\psi$ type anyons: an area in which they may be placed such that their state becomes delocalized along the line and they are no longer detected by the stabilizers \cite{raussendorf,fowler_hole,wootton_hole2}. Similar holes can also be engineered for the constituent charge and flux anyons. A defect line is therefore a special case of the combination of a charge and flux hole, in which only $\psi_g = e_g \times m_g$ type anyons may reside rather than general $e_g \times m_h$ anyons. Nevertheless, we can consider a process in which a defect line is transformed into a separate charge and flux hole. Details on these holes and the transformations between them can be found in Appendix~\ref{app:holes}.

When only the charge hole of one qubit is braided around the defect line of another, the process for an initial state $\ket{g,h}$ corresponds to braiding an $e_g$ around a $\psi_h$, which would yield the phase $\omega^{gh}$. The charge and flux holes can then be recombined into a defect line. The net effect of the entire process is to apply the controlled phase gate $\Lambda$.
Such a process is illustrated in Fig.~\ref{fig:phasegate}.

\begin{figure}
\includegraphics[width=0.8\columnwidth]{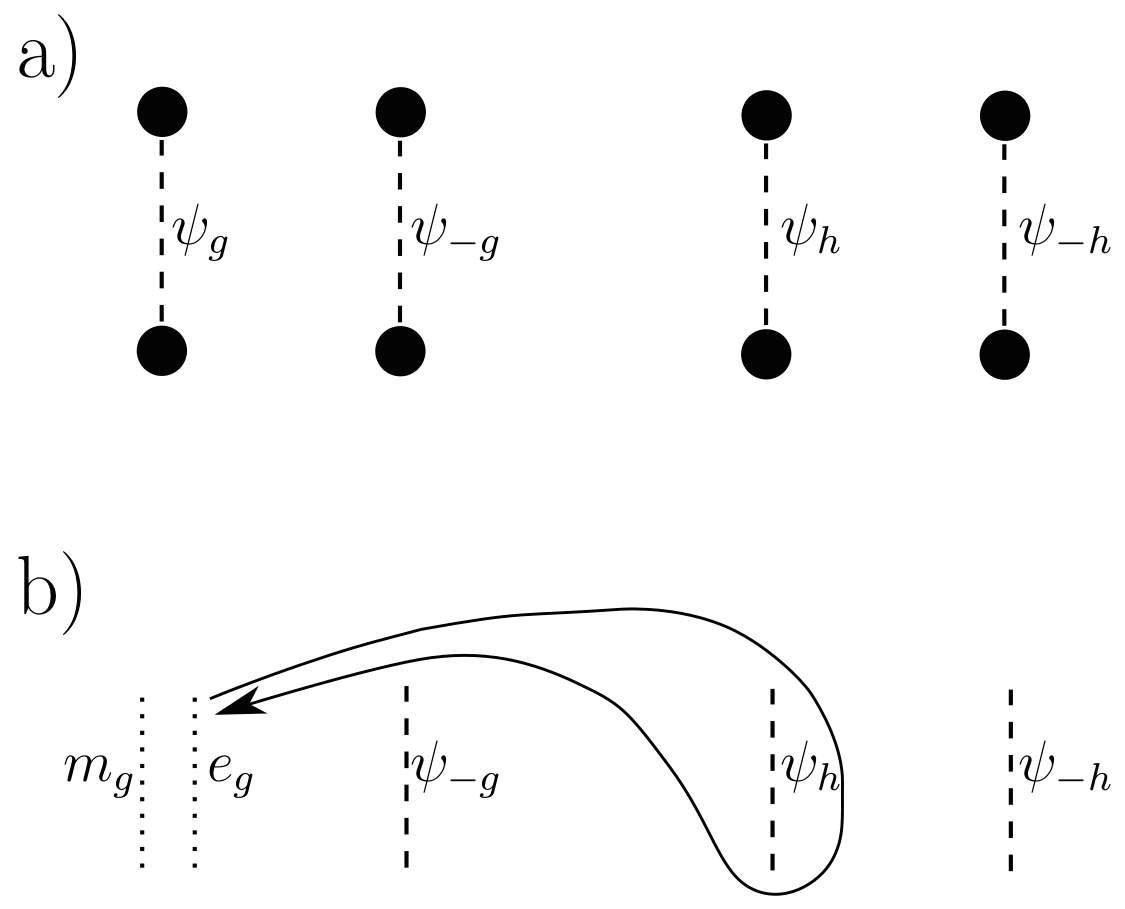}
  \caption{Performance of a controlled phase-gate. The a) part of the figure shows a logical product state $\ket{g,h}$ stored in the fusion space of eight parafermions. 
	   The defect line storing a mode $\psi_g$ can be split into two defect lines storing $D(\ZZ_4)$ charges $e_g$ and $m_g$, respectively. 
	  Braiding both endpoints of one of these lines clockwise around the defect line storing the $\psi_h$ mode, as shown in the b) part, produces a phase $\omg^{gh}$, as required.}
  \label{fig:phasegate}
\end{figure}

One process which could split the defect line in this way is simply to intersect it with two others. One would be a line along which charge anyons are hopped by high-strength terms. The other would similarly hop flux anyons. The stabilizers that detect charges and fluxes, respectively, along these lines would then be suppressed. By adiabatically removing the defect line which delocalizes $\psi$ modes, its $\psi_g$ anyon occupation would be transferred to these two lines. The recombination of the defect line would be done by the reverse process.

For a tensor product of $d$-level systems, the Clifford group $\mathcal{C}_d$ consists of gates that map tensor products of $d$-level Pauli operators to other such tensor products under conjugation.
In Appendix~\ref{app:clifford}, we prove the following theorem.

\vspace{2mm}
\noindent\textbf{Theorem.}
\textit{The single-qudit gates $S$, $T$, and $Z$, and nearest-neighbor controlled phase-gates $\Lambda$ generate the entire Clifford group $\mathcal{C}_4$.}
\vspace{2mm}

As an example, $\tilde{H}=STS=TST$ satisfies $\tilde{H}X\tilde{H}\mdag=Z$ and $\tilde{H}Z\tilde{H}\mdag=X\mdag$, so it can be identified with the logical Hadamard gate, up to a phase. 
Indeed, using the standard definition
\begin{align}
 H=\frac{1}{\sqrt{d}}\sum_{gh}\omg^{gh}\ket{g}\bra{h}\,,
\end{align}
one verifies that $\sqrt{\omg}H=\tilde{H}$.

One possible implementation of $H$ (up to a phase) is a cyclic permutation of the four parafermion modes, as in Fig.~\ref{fig:braiding}.
This can be pictorially understood as follows. 
An $X$ corresponds to a transfer of a $\psi_1$ from the right to the left defect line, accompanied by a clockwise loop of a $r_1$ around a horizontal pair.
A $Z$ corresponds to a clockwise loop of a $r_1$ around a vertical pair.
A $\pi/2$ rotation as performed by $H$ thus maps these two operations onto each other, up to the fact that we do not perform a vertical $\psi_1$ transfer, as the $\psi$ occupancy of the vertical pair is delocalized along the defect line.

\section{Error correction}\label{sec:errcorr}

For any system with a finite energy gap at finite temperature, excitations will appear with a finite density. This corresponds to finite length scale on which quantum computation can be performed before errors are almost certain to appear. This length scale can be increased by increasing the gap or lowering the temperature. However, neither of these methods is truly scalable. Error correction is therefore required if scalable quantum computation is to be performed.

For non-Abelian systems, the first studies of the corresponding error correction problem have recently appeared \cite{brell,woottonNA,hutter,wootton_hutter,burton,hutter_proof}. Error correction for non-Abelian anyons is still poorly understood and its feasibility has not been demonstrated for the (realistic) time-continuous case. It comes thus very welcome that while our system provides the computational power of non-Abelian parafermions, its physical excitations still are Abelian $D(\ZZ_4)$ anyons, and the error correction problem for $D(\ZZ_n)$ quantum double models (including the time-continuous case) is well-studied \cite{duclos,anwar,hutter,watson_double,wootton_hutter}. However, when correcting these $D(\ZZ_4)$ anyons, we face a number of difficulties not considered in previous studies \cite{duclos,anwar,hutter,watson_double}:
\begin{itemize}
 \item[(i)] Our stabilizer operators are products of $\ZZ_4$ qudit operators $X$, $X\mdag$, $Z$, $Z\mdag$, while an error model is realistically expressed in terms of single-qubit operators $\px$, $\py$, $\pz$. These do not map eigenstates of the stabilizer operators to other eigenstates and one single-qubit operator can produce a product of up to three qudit operators (see below).
 \item[(ii)] We consider quantum information stored in a synthetic topological degeneracy, which involves a defect line allowing anyons to change from one sublattice to the other (stars to hexagons and \emph{vice versa}). We thus cannot decode each sublattice separately, as usually done for the toric code and other $D(\ZZ_d)$ quantum double models, but have to correct both of them simultaneously while taking the possibility of transferring anyons from one to the other into account.
 \item[(iii)] Besides simplistic i.i.d.\ error models (such as depolarizing noise), we are particularly interested in Hamiltonian protection of a quantum state subject to thermal errors.
 \item[(iv)] We do not consider a square lattice, but a trihexagonal one, which makes moving anyons and defining their distance more involved.
\end{itemize}

\subsection{Error model}

Since our $4$-level qudits are composed of two qubits, it is natural to consider an error model in terms of single-qubit operations $\px$, $\py$, and $\pz$.
For a qudit hosted in two qubits $1$ and $2$, single-qubit Pauli operators can be expressed in terms of $\ZZ_4$ operators by inverting Eq.~\ref{eq:XYZ}. We find
\begin{align}
 \px_1 &= \half X(1-Z^2) +\hc\,, \nn\\ \px_2 &= \half X(1+Z^2) +\hc\,, \nn\\
 \py_1 &= \half e^{i5\pi/4} Y(1+Z^2) +\hc\,, \nn\\ \py_2 &= \half e^{i3\pi/4}Y(1-Z^2) +\hc\,, \nn\\
 \pz_1 &= e^{-i\pi/4}Z +\hc\,,  \nn\\ \pz_2 &= e^{i\pi/4}Z +\hc\,.
\end{align}
If we start from an eigenstate of all stabilizer operators, applying single-qubit Pauli operators will generate a superposition of states corresponding to different syndrome outcomes.
By measuring all stabilizer operators, we can project again into a subspace with definite syndrome values. Each single-qubit Pauli operator thereby translates into a product of up to three qudit operators.
Table~\ref{tab:conversion} summarizes (up to irrelevant phases) into which qudit operators a certain single-qubit Pauli operator will translate with equal probability.

\begin{table}
\begin{tabular}{|l|l|}
 \hline
  $\px_1$, $\px_2$ & $X$, $X\mdag$, $XZ^2$, $X\mdag Z^2$ \nn\\
 \hline
  $\py_1$, $\py_2$ & $XZ$, $X\mdag Z$, $XZ\mdag$, $X\mdag Z\mdag$ \nn\\
 \hline
 $\pz_1$, $\pz_2$ & $Z$, $Z\mdag$ \nn\\
 \hline
\end{tabular}
\caption{Conversion from single-qubit Pauli operators to $4$-dimensional generalized Pauli operators. 
When a syndrome measurement is performed, a Pauli operator is converted to each of the generalized Pauli operators in the right-hand column with equal probability.}
\label{tab:conversion}
\end{table}

As a first simple error model, which does not involve a notion of Hamiltonian protection, we consider depolarizing noise. That is, for each qubit of the code we apply a Pauli operator with some probability $p$ (the \emph{depolarization rate}), where each of the three Pauli operators is chosen with equal probability.

More interesting from a physical perspective is a thermal error model. We consider a quantum state stored in the degenerate groundstates of the Hamiltonian given in Eq.~(\ref{eq:heff}), and assume that the system is weakly coupled to a heat bath at some temperature $T$. Following e.g. Ref.~\cite{brell}, we assume that evolving the system according to the Metropolis algorithm provides a reasonable approximation of the thermalization process, 
since the evolution obtained by means of the Metropolis algorithm is local, Markovian, and has the thermal state as its unique fixed point.

During our simulation, we proceed as follows. We first pick one of the spins-$\half$ of the system at random, then pick one of the three single-qubit opertors acting on that qubit at random, 
and convert that to a $4$-dimensional generalized Pauli operator according to Table~\ref{tab:conversion}. We then calculate the energy cost $\Delta_{\text{tot}}$ of applying that generalized Pauli operator (or products thereof).
This energy cost is of the form 
\begin{align}
 \Delta_{\text{tot}} = m\Delta_{\triangle} + n\Delta_{\hexagon}\,,
\end{align}
where $\Delta_{\triangle}$ and $\Delta_{\hexagon}$ are the energy costs of creating a single triangle/hexagon-type anyon with charge $1$ or $3$ in Eq.~(\ref{eq:heff}), respectively.
(That is, $\Delta_{\triangle}=2J$ and $\Delta_{\hexagon}=2\frac{63}{8}\frac{h^6}{(2J)^5}$.) 
Creating an anyon with charge $2$ will have an energy cost $2\Delta_{\triangle}$ or $2\Delta_{\hexagon}$. 
The coefficients $m$ and $n$ are elements of $\lbrace0,\pm2,\pm4\rbrace$, depending on the change in anyonic charge.
The proposed error is then accepted with probability $\min\lbrace1, e^{-\Delta_{\text{tot}}/k_BT}\rbrace$. 
The noise model we apply is thus the standard classical Metropolis algorithm that maps eigenstates of Eq.~(\ref{eq:heff}) to other eigenstates.
At any given time during our simulation, the system is ``classical'' in the sense that it does not involve superpositions of different anyon configurations.

If the proposed error is accepted, we copy the current state of the system and try to correct it. If correction is successful, we continue our simulation with the uncorrected version of the system. If correction fails (for at least one logical operator), we interpret this as the quantum information having survived for a time which is given by the number of Metropolis steps divided by the number of spins in the system.

The thermal error model has three relevant energy scales $k_BT$, $\Delta_{\hexagon}$, and $\Delta_{\triangle}$.
Since $\Delta_{\hexagon}$ appears in higher-order perturbation theory than $\Delta_{\triangle}$, we expect $\Delta_{\triangle}>\Delta_{\hexagon}$. Furthermore, effective protection requires $k_BT<\Delta_{\hexagon},\Delta_{\triangle}$.
We introduce a parameter $\lambda$ which quantifies the separation of these three energy scales, i.e., $\Delta_{\hexagon}=\lambda k_BT$ and $\Delta_{\triangle}=\lambda^2 k_BT$.
Very high values of $\lambda$ are uninteresting, since they exponentially suppress errors from occurring.

\subsection{Without defects}

If there are no defect lines present, the anyonic charge of both types of anyons is conserved (modulo $4$), and they can be corrected separately.
Various techniques have been developed for correcting general $D(\ZZ_n)$ quantum double models \cite{duclos,anwar,hutter,watson_double}.
However, correcting the $D(\ZZ_4)$ case is particularly easy, since we can exploit the relation $\ZZ_4/\ZZ_2\simeq\ZZ_2$. Specifically, we can first fuse all oddly-charged anyons in pairs. In a second round, the remaining anyons, which are all of charge $2$, are fused in pairs.
In order to find these pairings, we use the library \texttt{Blossom V} \cite{kolmogorov}, which is the latest implementation of the efficient minimum-weight perfect matching algorithm due to Edmonds \cite{edmonds}.
The weight between two equal-type anyons is thereby defined as the minimal number of generalized Pauli operators that need to be applied to create a pair of anyons at the two given locations.

\begin{figure}
 \includegraphics[width=.8\columnwidth]{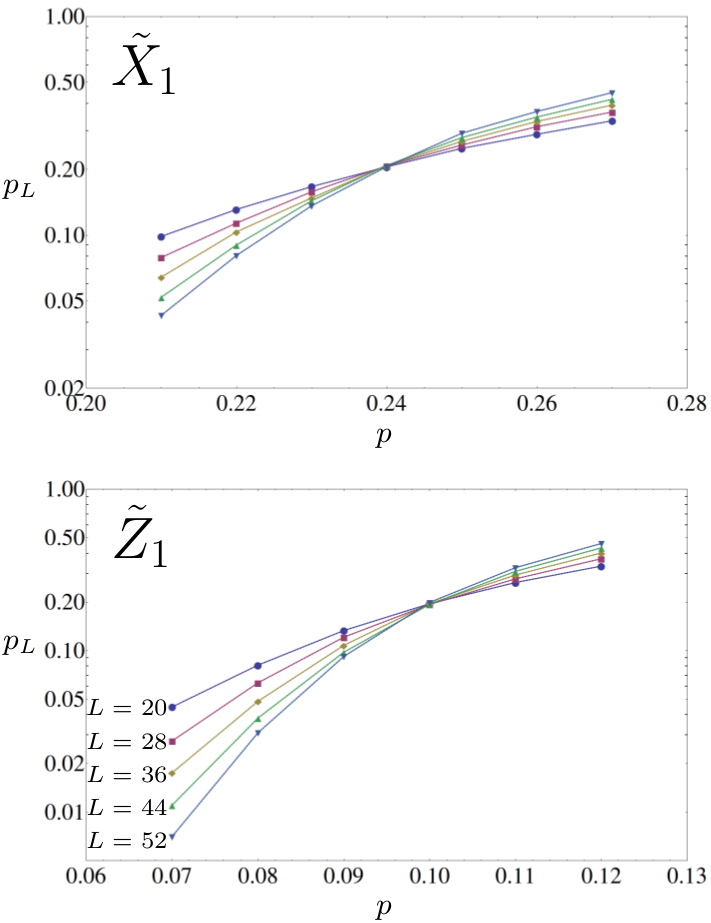}
  \caption{Error rates $p_L$ of the logical operators $\tilde{X}_1$ and $\tilde{Z}_1$ illustrated in Fig.~\ref{fig:operators} as a function of the qubit depolarization rate $p$ for code sizes $L=20,28,36,44,52$. Each data point represents $10^4$ logical errors, such that error bars are negligible. 
  We recognize a threshold error rate $p_c\approx24\%$ for $\tilde{X}_1$ and $p_c\approx10\%$ for $\tilde{Z}_1$.}
  \label{fig:depolThreshold}
\end{figure}

Fig.~\ref{fig:depolThreshold} shows our results for the depolarizing noise model, i.e., the logical error rates of the the logical operators $\tilde{X}_1$ and $\tilde{Z}_1$ illustrated in Fig.~\ref{fig:operators} as a function of the depolarization rate $p$. 
One clearly recognizes threshold error rates $p_c\approx24\%$ and $p_c\approx10\%$, respectively. The equivalent figures for the logical operators $\tilde{Z}_2$ and $\tilde{X}_2$ look very similar and yield equivalent threshold error rates $p_c$.

These thresholds are best compared with those for an equivalent code based on $\ZZ_2$ anyons, and so with only a single qubit on each vertex. For independent bit and phase flips, the thresholds for $\tilde{X}_1$ and $\tilde{Z}_1$ are $p_c\approx 16.4\%$ and $p_c\approx 6.7\%$, respectively \cite{beat,abbas}. When the hexagonal and triangular plaquettes are decoded separately, these correspond to thresholds of $p_c\approx 24.6\%$ and $p_c\approx 10.5\%$ for depolarizing noise. The similarity of these $\ZZ_2$ values with those of $\ZZ_4$ is remarkable. This qudit code is therefore just as adept at suppressing qubit noise as its qubit counterpart.

\begin{figure}
 \includegraphics[width=.8\columnwidth]{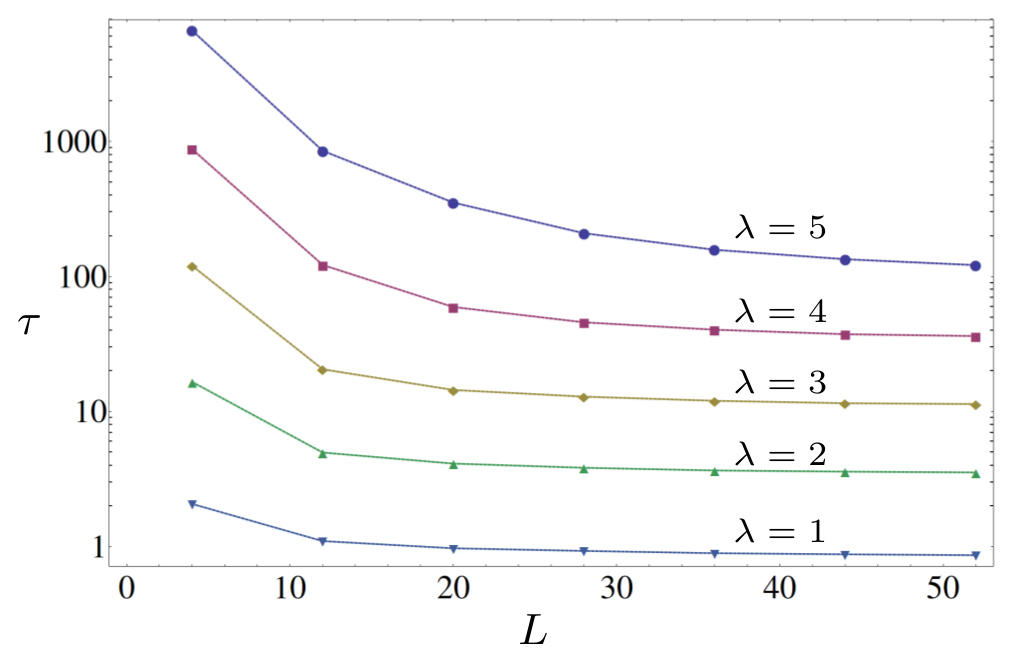}
  \caption{Average lifetimes $\tau$ of the logical qudit with logical operators $\tilde{X}_1$ and $\tilde{Z}_1$ as a function of code size $L$ for $\lambda=1,2,3,4,5$. Each data point represents $10^4$ experiments. 
  The lifetime is defined as the number of Metropolis steps until the first logical operator detects an error, divided by the number of spins in the code.}
  \label{fig:thermalLifetimes}
\end{figure}

It is well-known that the finite-temperature lifetime of a two-dimensional quantum memory with local interactions only is upper-bounded by a constant independent of the system size, see e.g.\ Ref.~\cite{brown}.
Fig.~\ref{fig:thermalLifetimes} shows the lifetime of a logical qudit with logical operators $X_1$ and $Z_1$ subject to the thermal error model. We notice lifetimes that decrease to an asymptotic value for large $L$ and considerable finite-size tails.
These tails correspond to the regime in which the breakdown of error correction is not due to the density of anyons becoming so high that pairing them becomes ambiguous, 
but where the breakdown is caused by one of the first pairs wandering along a topologically non-trival path around the torus. 
The smaller the system, the longer it takes to produce an anyon pair, leading to the observed tails for small enough $L$ and $T$ (large enough $\lambda$).

\subsection{With defects}

When defect lines as in Fig.~\ref{fig:defects} are present, the error correction problem becomes more involved. 
It is no longer possible to correct the two anyon types (hexagons and triangles in our case) separately, as is usually done for the $D(\ZZ_n)$ models \cite{duclos,anwar,hutter,watson_double}. 
Instead, error correction needs to take the possibility of converting between different anyon types into account. 
We thus pair all oddly-charge anyons of both types in a first round and all remaining charge $2$ anyons of both types in a second round.
Pairings can involve anyons which are of equal or of different type.
The weight for connecting two anyons is defined as the minimal number of generalized Pauli operators needed to create a pair of anyons at their respective positions from the vacuum. 
For equal-type anyons, this will be an error string that crosses an even number of defect lines, while for different-type anyons this will be an error string that crosses an odd number of defect lines. 
This can mean, for instance, that connecting two equal-type anyons can have a large weight despite them being geometrically nearby, if there is a defect line between them.

For a code of linear size $L$ in both dimensions, with periodic boundary conditions and $L$ even, we choose defect lines involving $L/2+1$ qudits, as shown in Fig.~\ref{fig:defects} for $L=20$. 
The logical operators $\tilde{X}_L$ and $\tilde{Z}_L$ then have a distance $L+2$ and $L/2+4$, respectively.

\begin{figure}
 \includegraphics[width=.8\columnwidth]{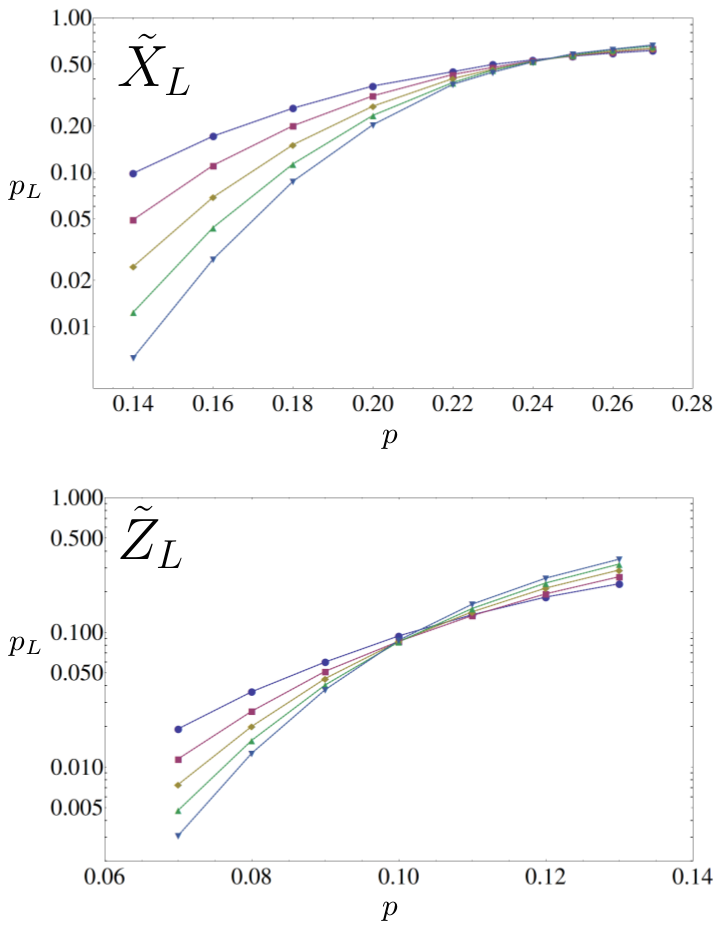}
  \caption{Error rates $p_L$ of the logical operators $\tilde{X}_L$ and $\tilde{Z}_L$ illustrated in Fig.~\ref{fig:defects} as a function of the qubit depolarization rate $p$ for code sizes $L=20,28,36,44,52$. Each data point represents $10^4$ logical errors, such that error bars are negligible. 
  We recognize a threshold error rate $p_c\approx24\%$ for $\tilde{X}_L$ and $p_c\approx10\%$ for $\tilde{Z}_L$.}
  \label{fig:depolThresholdDefect}
\end{figure}

For the depolarizing error model, we find the threshold error rates $p_c$ for both of the logical operators $\tilde{X}_L$ and $\tilde{Z}_L$ given in Fig.~\ref{fig:defects}. 
The results are given in Fig.~\ref{fig:depolThresholdDefect}. For the defect operator $\tilde{X}_L$, we find a threshold error rate $p_c\approx24\%$, as for the operators $\tilde{X}_1$ and $\tilde{X}_2$ in the defect-free case (Figs.~\ref{fig:operators} and \ref{fig:depolThreshold}), while for the defect operator $\tilde{Z}_L$ we find a threshold error rate $p_c\approx10\%$, as for the operators $\tilde{Z}_1$ and $\tilde{Z}_2$ in the defect-free case.

The fact that these values coincide with the defect free case is not unexpected. The introduction of the defects essentially corresponds to a change in the boundary conditions. However, the vast majority of errors have large support within the bulk. The value of the threshold is therefore dominated by bulk effects rather than boundary effects.

\begin{figure}
 \includegraphics[width=.8\columnwidth]{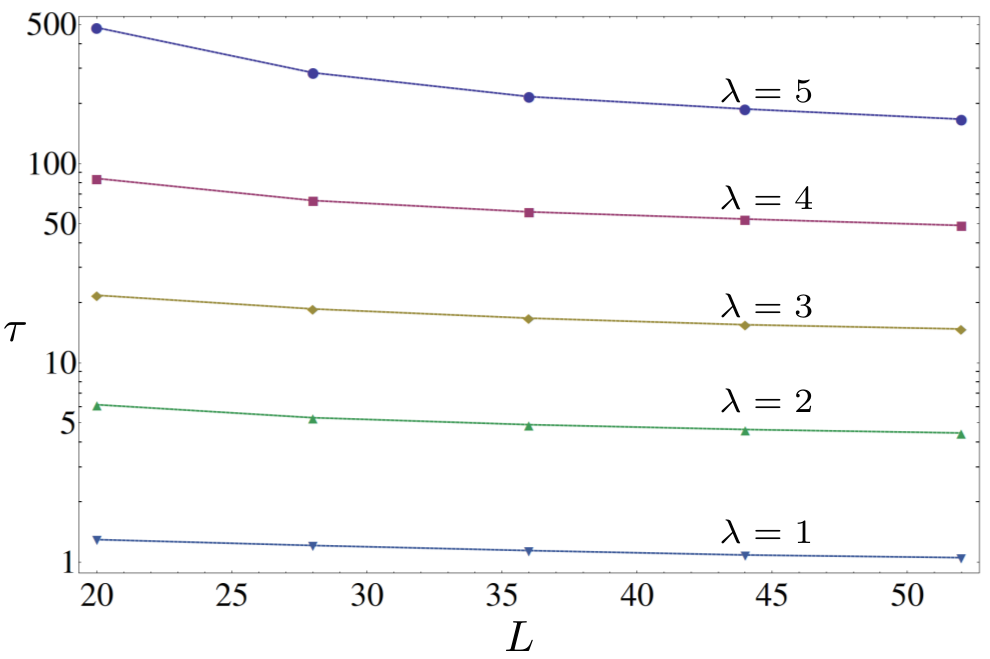}
  \caption{Average lifetimes $\tau$ of the logical qudit stored in the defect logical operators $\tilde{X}$ and $\tilde{Z}$ illustrated in Fig.~\ref{fig:defects} as a function of code size $L$ for $\lambda=1,2,3,4,5$. Each data point represents $10^4$ experiments. The lifetime is defined as the number of Metropolis steps divided by the number of spins in the code.}
  \label{fig:thermalLifetimesDefect}
\end{figure}

Fig.~\ref{fig:thermalLifetimesDefect} shows the average lifetime of the qudit stored in the synthetic topological degeneracy in Fig.~\ref{fig:defects} for the thermal error model. We note that for a given parameter $\lambda$, the asymptotic lifetimes ($L\ra\infty$) are close to those in the defect-free case given in Fig.~\ref{fig:thermalLifetimes}.

\section{Conclusions}\label{sec:conclusions}

We have proposed a system which, on the physical level, involves only nearest-neighbor two-qubit interactions, allows one to perform all Clifford gates through quasi-particle braiding, and has a well-understood error correction problem. 

We have greatly benefitted from the fact that our non-Abelian system is built on top of a system whose excitations correspond to an Abelian anyon model. This allows us to perform the logical operators $X$, $Z$, and $\Lambda$ through quasi-particle braiding. 
It also makes our error correction problem manageable, despite some subtleties such as the fact that single-spin Pauli operators generate superpositions between different syndrome outcomes and the ability to convert between different anyon species during error correction.

Universal quantum computation requires the ability to perform non-Clifford gates, such as ``small-angle'' unitaries. While it is not difficult to perform a non-Clifford operation by non-topological means in our system, this abandons fault-tolerance.
The technique of magic state-distillation \cite{bravyi_magic} is typically used to restore fault-tolerance. 
While research on magic state distillation has so far focused on prime qudit dimensions $d$ \cite{campbell_magic,campbell}, qudit codes with the right transversality properties to perform magic state distillation in non-prime dimensions, including $d=4$, also exist \cite{watson}.
Unfortunately, for non-prime $d$ it is not known whether Clifford gates plus an arbitrary non-Clifford gate are sufficient to achieve universality \cite{campbell_commu}.
It is our hope that our work fuels interest in the $d=4$ case, being a power of $2$ and thus allowing to employ qubits, as demonstrated in our work, while being the smallest power of $2$ that allows one to go beyond the Ising/Majorana case.

Alternatively, one could imagine energetically penalizing one of the degrees of freedom of a two-qubit Hilbert space to obtain a synthetic qutrit ($d=3$).
Magic state distillation for qutrits is well-studied \cite{campbell_qutrit}, potentially allowing to perform fault-tolerant universal quantum computation with $\ZZ_3$ parafermions in a qubit system \cite{comment}.

The authors thank M.~Barkeshli for elaborations on the development of the idea of generating non-Abelian defects in topological systems.
This work was supported by the SNF, NCCR QSIT, and IARPA.

\newpage\quad\newpage
\appendix

\section{Sixth-order degenerate perturbation theory}\label{app:perturbation}
For our perturbation theory, we employ a Schrieffer-Wolff transformation \cite{schrieffer}, as formalized in Ref.~\cite{bravyiSW}.

Consider an unperturbed Hamiltonian $H_0$ whose spectrum can be separated into a low- and a high-energy subspace, which are energetically separated by a gap. 
Given a perturbation $V$, we want to find an effective Hamiltonian $\heff$ describing the ``effective'' physics on the low-energy subspace.
The effective Hamiltonian can be developed in a perturbative series
\begin{align}
 \heff = \heff^{(0)} + \heff^{(1)} + \heff^{(2)} + \ldots
\end{align}
in powers of some small expansion parameter.

Let $P$ denote the projector onto the low-energy subspace and $Q=\id-P$ the projector onto the high-energy subspace.
We define $\vd=PVP+QVQ$ and $\vod=V-\vd=PVQ+QVP$.
For some operator $A$, we define the superoperator $\hat{A}$ via $\hat{A}(O)=[A,O]$.
Let $H_0=\sum_iE_i|i\rangle\langle i|$ be the spectral decomposition of $H_0$ and define the superoperator $\LL$ via
\begin{align}
 \LL(O) = \sum_{i,j}\frac{\langle i|QOP|j\rangle}{E_i-E_j}|i\rangle\langle j| - \hc\,.
\end{align}
We employ the convention that unless indicated otherwise by use of brackets, a superoperator $\LL$ acts on all operators to its right.

For the sixth-order effective Hamiltonian, one derives from Ref.~\cite{bravyiSW} the expression
\begin{widetext}
\begin{align}\label{eq:heff6}
 \heff^{(6)} = \half P\hat{S}_5(\vod)P -\frac{1}{24}P(\hat{S}_1^2\hat{S}_3+\hat{S}_1\hat{S}_3\hat{S}_1+\hat{S}_3\hat{S}_1^2+\hat{S}_2^2\hat{S}_1+\hat{S}_2\hat{S}_1\hat{S}_2+\hat{S}_1\hat{S}_2^2)(\vod)P + \frac{1}{240}P\hat{S}_1^5(\vod)P\,,
\end{align}
\end{widetext}
where
\begin{align}\label{eq:S}
 S_1 &= \LL(\vod) \nn\\
 S_2 &= -\LL\hvd(S_1) \nn\\
 S_3 &= -\LL\hvd(S_2)+\frac{1}{3}\LL\hat{S}_1^3(\vod) \nn\\
 S_4 &= -\LL\hvd(S_3)+\frac{1}{3}\LL(\hat{S}_1\hat{S}_2+\hat{S}_2\hat{S}_1)(\vod) \nn\\
 S_5 &= -\LL\hvd(S_4)+\frac{1}{3}(\hat{S}_2^2+\hat{S}_1\hat{S}_3+\hat{S}_3\hat{S}_1)(\vod) \nn\\&\quad -\frac{1}{45}\LL\hat{S}_1^4(\vod)\,.
\end{align}

In our case, the low-energy subspace onto which $P$ projects is given by the space in which all triangle operators in Eq.~(\ref{eq:paraham}) have minimal energy, i.e., $Z_aZ_bZ_c\equiv1$ for all triangles $(a,b,c)$.
This subspace is fully degenerate.
The lowest-energetic excitations change the eigenvalue of a stabilizer $Z_aZ_bZ_c$ from $1$ to $\pm i$. Since the eigenvalue of $-(Z_aZ_bZ_c+\hc)$ is thereby changed from $-2$ to $0$, this has an energy cost $\Delta=2J$. Note, however, that stabilizer eigenvalues can only be changed in pairs, such that the gap between the low-energetic (groundstate) subspace and the space of excited states is in fact given by $2\Delta$.

A crucial property of our Hamiltonian is that there is no lower-than-sixth-order perturbation that acts within the groundstate space.
Therefore, we are only interested in terms of the form $P\vod(\vd)^4\vod P$, which allows to greatly simplify the effective Hamiltonian.
Namely, only the first summand in all expressions in Eqs.~(\ref{eq:heff6}) and (\ref{eq:S}) is relevant in our case.
We find
\begin{align}
 \heff^{(6)} = \half P\left[(\LL\hvd)^4(\LL\vod),\vod\right]P\,.
\end{align}
Using now that in our case $\vd P=0$, this can be further simplified to
\begin{align}
 \heff^{(6)} &= \half P\LL\left(\LL\left(\LL\left(\LL\left(\LL\left(\vod\right)\vd\right)\vd\right)\vd\right)\vd\right)\vod P \nn\\
&\quad -\half P\vod(\LL\vd)^4(\LL\vod)P \nn\\
&= -P\vod(\LL\vd)^4(\LL\vod)P\,.
\end{align}

There are $6!=720$ possibilities for applying the six factors $X_rX\mdag_sX_tX\mdag_uX_vX\mdag_w$ around one hexagon which leads the system back to the groundstate.
Table~\ref{tab:processes} lists all possible routes the excitation energy above the groundstate can take, together with their numbers of possibilities. 

\begin{table}
\begin{tabular}{|l|l|}
 \hline
 $0\ra2\Delta\ra2\Delta\ra2\Delta\ra2\Delta\ra2\Delta\ra0$ & 96 \\ 
 \hline
 $0\ra2\Delta\ra4\Delta\ra2\Delta\ra2\Delta\ra2\Delta\ra0$ & 48 \\ 
 \hline
 $0\ra2\Delta\ra2\Delta\ra4\Delta\ra2\Delta\ra2\Delta\ra0$ & 48 \\ 
 \hline
 $0\ra2\Delta\ra2\Delta\ra2\Delta\ra4\Delta\ra2\Delta\ra0$ & 48 \\ 
 \hline
 $0\ra2\Delta\ra4\Delta\ra4\Delta\ra2\Delta\ra2\Delta\ra0$ & 96 \\ 
 \hline
 $0\ra2\Delta\ra2\Delta\ra4\Delta\ra4\Delta\ra2\Delta\ra0$ & 96 \\ 
 \hline
 $0\ra2\Delta\ra4\Delta\ra4\Delta\ra4\Delta\ra2\Delta\ra0$ & 192 \\ 
 \hline
 $0\ra2\Delta\ra4\Delta\ra2\Delta\ra4\Delta\ra2\Delta\ra0$ & 24 \\ 
 \hline
 $0\ra2\Delta\ra4\Delta\ra6\Delta\ra4\Delta\ra2\Delta\ra0$ & 72 \\ 
 \hline
\end{tabular}
\caption{Possible routes the excitation energy above the groundstate can take (left column), together with their respectiv multiplicities (right column). Note that the number of multiplicities adds up to $6!=720$.}
\label{tab:processes}
\end{table}

In conclusion, we find the sixth-order effective Hamiltonian
\begin{align}
 H_{\text{eff}} = -q\frac{h^6}{\Delta^5}(X_rX\mdag_sX_tX\mdag_uX_vX\mdag_w+\hc)\,,
\end{align}
where the dimensionless prefactor
\begin{align}
 q &= \frac{96}{32}+\frac{48}{64}+\frac{48}{64}+\frac{48}{64}+\frac{96}{128}+\frac{96}{128} \nn\\&\quad +\frac{192}{256}+\frac{24}{128}+\frac{72}{384} \nn\\& = \frac{63}{8}
\end{align}
is given by the multiplicities in Table~\ref{tab:processes}, divided by the product of all excitation energies (in multiples of $\Delta$) along the virtual process.

\section{Moving unpaired parafermion modes}\label{app:braiding}

To consider the creation and braiding of unpaired parafermionic modes, we must first decide on the double plaquettes with which we will work. Let us consider those of Fig.~\ref{fig:braid}. To visualize the two parafermion modes within each double plaquette we use light blue circles. The one to the right of a double plaquette $P$ is labelled $P_1$, and that to the left is $P_2$.

Parity operators for $\psi$  modes are defined on pairs of parafermion modes. We are primarily concerned with two types of pairing: those of the two modes within the same double plaquette, and those of two modes from neighbouring double plaquettes. Relevant examples of the latter type are shown in Fig.~\ref{fig:braid} by red, orange and green lines connecting the corresponding modes.

\begin{figure}
\centering
\includegraphics[width=1.0\columnwidth]{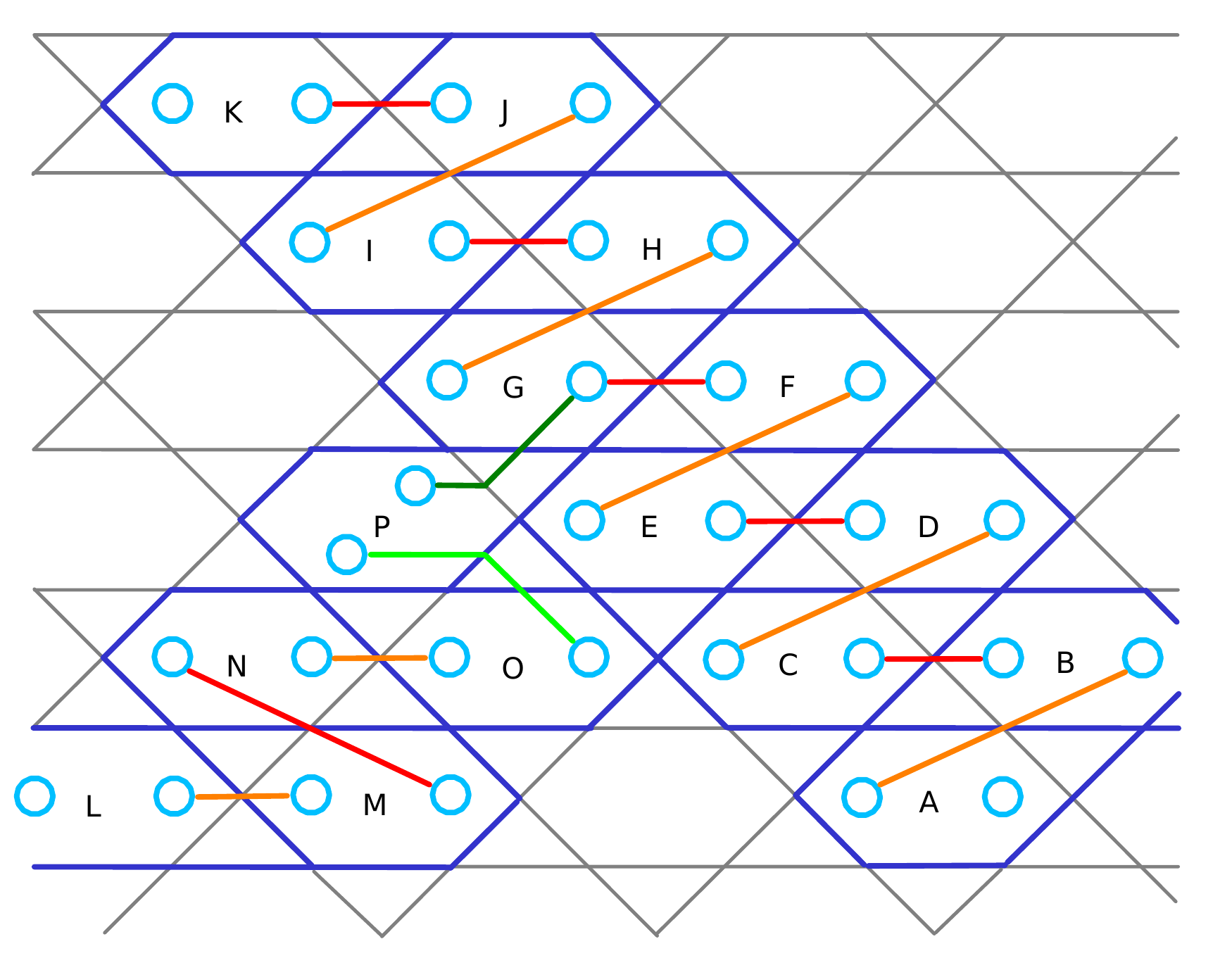}
\caption{A selection of double plaquettes used in a braiding operation. Each double plaquette corresponds to two parafermion modes. For a double plaquette $P$, the parafermion to the right is labelled $P_1$, and that to the left is $P_2$.}
\label{fig:braid}
\end{figure}

For the two modes within each double plaquette, the parity operator $\omega^{5/2} \gamma_{P_1} \gamma\mdag_{P_2}$ corresponds to the stabilizer $S_P$. The orange and red lines connecting modes $P_2$ to $(P+1)_1$ denote the parity operators $ \omega^{5/2} \gamma_{P_2} \gamma\mdag_{(P+1)_1}$. For orange lines, these correspond to the operator $Y$ on the vertex through which the line passes. For red lines they correspond to the operator $X\mdag Z$.

Consider a state initially within the stabilizer space. The parity operators for the pairs of parafermion modes within each double plaquette are therefore part of the stabilizer. The state therefore corresponds to this definite pairing of the modes.

Let us now consider the removal of the operator $S_A$ from the set of stabilizer generators (while $R_A$ remains). The corresponding  parafermion modes are now, in some sense, unpaired. This contributes a factor of four to the ground space degeneracy \cite{note}. However, due to the fact that the `unpaired' parafermions are not well separated, it is not difficult for local perturbations to lift the degeneracy of this space. To become truly unpaired, and benefit from topological protection, they must be separated.

To do this, we can add a term $K(Y+ Y \mdag)$ to the Hamiltonian, which corresponds to the parity operator $ \omega^{5/2} \gamma_{A_2} \gamma\mdag_{B_1}$. This acts on the vertex through which the orange line connecting these modes passes.

For $K \gg J$, this new term will overwhelm the $S_B$ term. The pairing of $B_1$ and $B_2$ will then be broken, and $B_1$ will become paired with $A_2$ instead. The unpaired mode originally at $A_2$ is therefore effectively moved to $B_2$. If the new term is introduced adiabatically, the degenerate subspace associated with the unpaired parafermion modes will remain in the same state during this process.

Similar processes can be used to move the unpaired modes further. The Hamiltonian term $K(X\mdag Z + X \mdag Z)$  corresponding to $ \omega^{5/2} \gamma_{B_2} \gamma\mdag_{A_1}$ can then be used to move the parafermion at $B_2$ to $C_2$, for example. Unpaired parafermion modes can therefore be separated by arbitrary distances, at the endpoints of lines on which single qudit terms are added to the Hamiltonian. In terms of qubits, these correspond to two-body interactions between qubits in the same site.

In order to unlock the potential of parafermions for quantum computation, it must be possible to braid the parafermion modes. Let us consider a specific example of this, using the system of Fig.~\ref{fig:braid}. Consider an initial state within the stabilizer space of all $S_P$ except $A$ and $L$. At these two points, we have the unpaired parafermion modes $A_1$, $A_2$, $L_1$ and $L_2$. Let us now consider operations such as those described above to move $A_1$ and $L_2$ away, beyond the bottom of the figure. All four parafermion modes are then well separated, and so the ground state degeneracy is topological protected.

We will now consider the exchange of $A_2$ with $L_1$. We do this by first moving $A_2$ to $K_2$, then $L_1$ to $A_2$, and finally $K_2$ to $L_1$. The two modes have then swapped places. An exchange of opposite chirality would correspond to first moving $L_1$ to $K_1$, and so on. Note that all modes are kept well separated during the exchange, and so topological protection is always maintained.

During the exchange, the movement of the modes is mostly achieved using Hamiltonian terms that correspond to the red and orange pairings in Fig.~\ref{fig:braid}. These are all single qudit terms. At the junction, however, terms corresponding to the green pairings are used. We must therefore consider these in detail.

For the pairing show by the light green line, the parity operator is $\omega^{5/2}\gamma_{P_2} \gamma\mdag_{O_1}$. This has the effect of creating an $\psi_g$, $\psi_{-g}$ pair on the double plaquettes $O$ and $P$. This requires the two qudit operator $X\mdag_3 X_4 Z\mdag_4$. For the dark green pairing, the $ \omega^{5/2} \gamma_{P_1} \gamma\mdag_{G_2}$ parity operator similarly requires the three qudit operator $\omega^{5/2}X_1 X\mdag_2 Z_2 Z_3$. These terms correspond to four- and six-body quasi-local interactions on the corresponding qubits, respectively. They can be realized by standard methods of perturbative gadgets. However, note that they need only be implemented while an exchange is in progress.

\section{Exchange of two parafermion modes}\label{app:exchange}

To determine the effects of a single clockwise exchange we consider the smallest possible implementation. This involves the double plaquettes labelled $F$, $G$ and $P$ in Fig.~\ref{fig:braid}, which are shown in more detail in Fig.~\ref{fig:braid2}. We consider a state in which the modes at $P_1$ and $F_2$ are unpaired, and those at $G_1$ and $G_2$ are paired by the Hamiltonian term $S_G$. Using this, we determine the effects of exchanging the unpaired modes.
The method used to exchange the two modes is similar to previous methods proposed in order to perform anyon braiding \cite{bonderson,bonderson_interactions}.

\begin{figure}
\centering
\includegraphics[width=1.0\columnwidth]{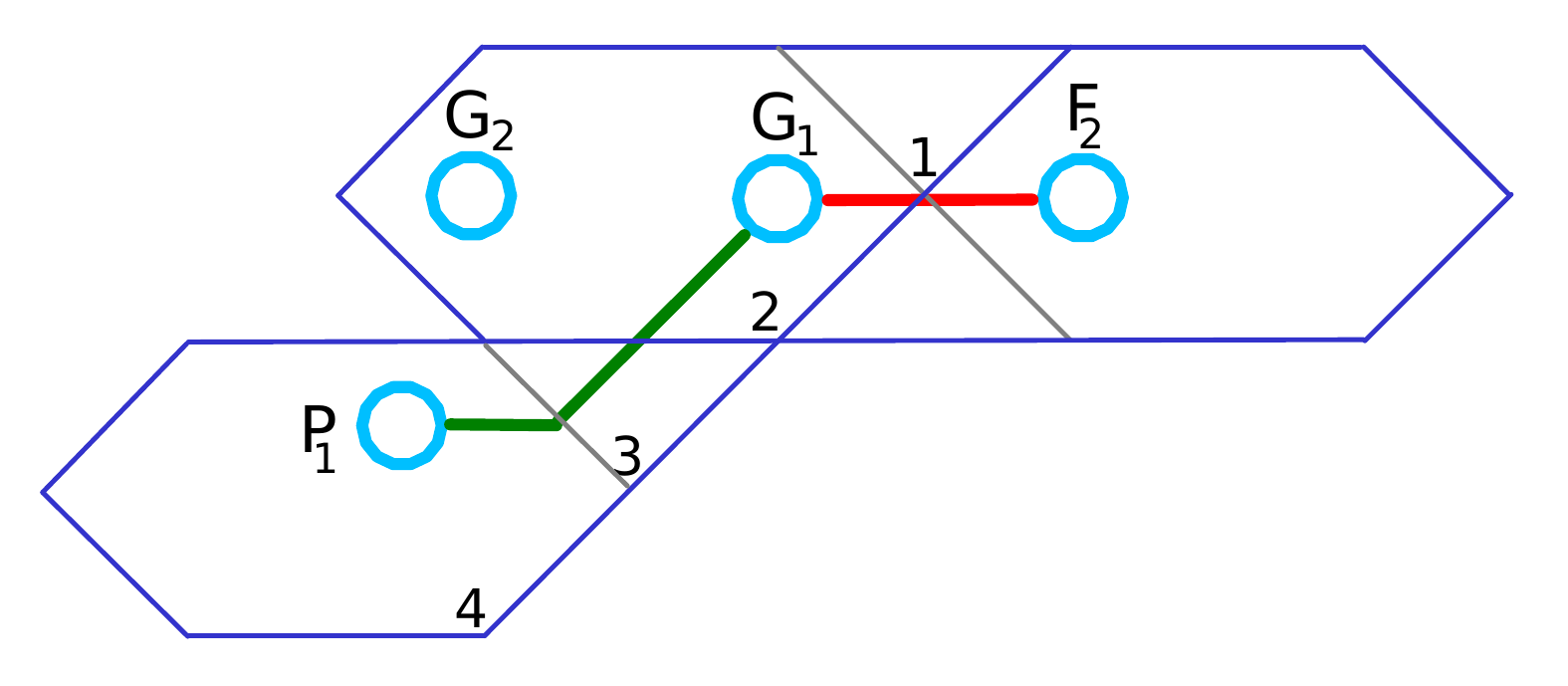}
\caption{Double plaquettes used in the worked example of a single exchange.}
\label{fig:braid2}
\end{figure}

The results of the exchange are most easily understood in terms of the $\psi$ mode formed by this pair. The parity operator, $\Gamma$ for this mode is an operator that creates a $\psi_1$, $\psi_{-1}$ pair and places them in double plaquettes $P$ and $F$, respectively. Also it must have eigenvalues of the form $\omega^g$, and so $\Gamma^4 = 1$. These conditions are satisfied by
\begin{align}
\Gamma = \omega^{(1+2a)/2} Z_1 X\mdag_2 Z_2 Z_3.
\end{align}
We similarly require operations that can move parafermions between the relevant plaquettes. These correspond to the green line between $G$ and $P$ and the red line between $F$ and $G$. These are
\begin{align}
\Pi = \omega^{(1+2b)/2} X_1 X\mdag_2 Z_2 Z_3, \,\,\,  \Phi = \omega^{(1+2c)/2}X\mdag_1 Z_1,
\end{align}
respectively. In these relations $a$, $b$ and $c$ are all elements of $\mathbb{Z}_4$.

We have freedom in choosing the values $a,b,c \in \mathbb{Z}_4$ for these relations. The corresponding freedom also exists for all operators used to move parafermion modes, as well as the logical operators. The values used do not simply correspond to differences in a global phase. Instead they determine which eigenspace of these operators has eigenvalue $\omega^0 = 1$, and so which one corresponds to the vacuum occupancy $\psi_0$ of the $\psi$ mode. These phases therefore cannot be chosen entirely arbitrarily, since the overall conservation constraint of $\psi$ modes must be maintained. However, since here we do not explicitly consider the operations that placed unpaired parafermion modes at $P_1$ and $F_2$, we can assume that their phases are chosen in a way that maintains this conservation. We will therefore consider a free choice of $a$, $b$, and $c$.

The first step in exchanging the parafermions is to move the one at $P_1$ to $G_2$. This is done by adiabatically changing the Hamitonian to one in which the term $\Pi+\Pi\mdag$ is present and stronger than $S_G$. This causes the modes at $G_1$ and $P_1$ to pair, moving the mode once at $P_1$ to $G_2$. The mode at $F_2$ is then moved to that at $P_1$ by adiabatically changing the Hamitonian to one in which the $\Phi+\Phi\mdag$ term is present and stronger than $S_G$, and the $\Pi+\Pi\mdag$ term is removed, pairing $F_2$ with $G_1$. The mode at $G_2$ is then moved to $F_2$ by adiabatically removing the $\Pi+\Pi\mdag$ term and so allowing $S_G$ to become dominant and $G_1$ and $G_2$ to pair. This process then results in the clockwise exchange of the modes.

The first step of this transformation takes a state that is initially in the $\omega^0$ eigenspace of $S_G$ and projects it to one in the $\omega^0$ eigenspace of $\Pi$. The next step projects the state into the $\omega^0$ eigenspace of $\Phi$. The final step projects back into the $\omega^0$ eigenspace of $S_G$. The end effect is then $P_G P_\Phi P_\Pi P_G.$ Here $P_G$ is the projector onto the $\omega^0$ eigenspace of $S_G$, etc. The rightmost $P_G$ simply reflects the fact that the initial state lies within this eigenspace.

The parity operator $\Gamma$ for the pair of unpaired modes commutes with $S_G$, and so can be mutually diagonalized with the above operator. We can therefore interpret its effects in terms of the phase factor assigned to each of the possible $\psi_g$ eigenspaces of the $\psi$ mode of the pair. 

When doing this, different values of $a$, $b$, and $c$ will result in different operations. This may seem to contradict the standard notion of a topological protected operation. However, these differences can be most easily understood by considering movement of parafermion modes implemented by measurement rather than adiabatic Hamiltonian manipulation. This method forces pairing of parafermion modes by measuring the occupancy of their corresponding $\psi$ mode, and so forcing it to have a definite value. Ideally, this measurement will give the vacuum result $\psi_0$. The effect is then the same as the adiabatic manipulation. If a different $\psi_g$ results, it must be removed by fusing it with the unpaired parafermion mode being moved. The different values used for the phases when moving unpaired modes, such as $a$, $b$ and $c$ here, determine how the measurement results are interpreted in terms of $\psi$ anyons, and so determine the net $\psi_g$ fused with the modes being moved. As such, differences in the conventions used for an exchange will change the resulting operation only by a factor of $\Gamma^g$, for some value of $g$ that depends on $a$, $b$, and $c$.

For the standard convention used throughout this paper, with $a=b=c=2$, the phase assigned to a $\psi_g$ occupation by the exchange is $\omega^{g^2/2}\omega^{g(g+1)}$. These are indeed all square roots of $\omega^{g^2}$, as predicted in the main text. However, they are not of the elegant form $\omega^{g^2/2}$ that would be more conducive for the proofs of Appendix \ref{app:clifford}.

For an exchange operation that does have the required form, consider $a=1$ and $b=c=2$. The phase assigned to $\psi_g$ is $\omega^{-g^2/2}$ in this case, up to a global phase of $\omega^{1/2}$. The required phases $\omega^{g^2/2}$ would therefore be obtained from an anticlockwise exchange.

The fact that we obtain the phase $\omega^{-g^2/2}$ for a clockwise exchange in this case means we would get $\omega^{-g^2}$ for a full clockwise monodromy. This would seem to contradict the arguments of the main text, which predict a phase of $\omega^{g^2}$. However, note that these phases differ only by a factor of $\omega^{2g^2}=\omega^{2g}$, and so are equivalent up to a factor of $\Gamma^2$. Since such factors are to be expected for different choices of $a$, $b$ and $c$, this monodromy does not contradict our expectations.
The ambiguity in these phases reflects the fact that the anyon model described by the fusion rules in Eq.~(\ref{eq:fusion_rules}) obeys only projective non-Abelian statistics.

Note that the arguments above do not assume anything about the initial state of the exchanged parafermions. Only their initial positions and the operations used to move them are required. The effect of the braiding is expressed in terms of $\Gamma$, the parity operator for their shared $\psi$ mode, which can be defined for any pair of parafermions. It therefore does not matter what state the fusion space of the parafermions was initially in, and it does not matter whether or not they exist at the end of the same defect line. The effect of the braiding is the same in all cases.

\section{Generators of the Clifford group}\label{app:clifford}

For a tensor product of $n$ $d$-level systems $(\mathbb{C}^d)^{\otimes n}$, the Pauli group $\mathcal{P}_d$ is defined as the group generated by the generalized Pauli operators $X_i$ and $Z_i$, 
and the Clifford group $\mathcal{C}_d$ is defined as the normalizer of $\mathcal{P}_d$ in the unitary group on $(\mathbb{C}^d)^{\otimes n}$. 
That is, elements in $\mathcal{C}_d$ map tensor products of $d$-level Pauli operators to other such tensor products under conjugation.

We start with some general remarks on the action of $\cli$ on $\pa$ for $d=4$.
For $d=4$, an operator $X^aZ^b$ has eigenvalues $\lbrace1\rbrace$ if $a=b=0$, $\lbrace1,-1\rbrace$ if both $a$ and $b$ are even and at least one of them is non-zero, $\lbrace i^{1/2},i^{3/2},i^{5/2},i^{7/2}\rbrace$ if both $a$ and $b$ are odd, 
and $\lbrace 1,i,-1,-i\rbrace$ if $a+b$ is odd. Since the number of distinct eigenvalues is preserved under conjugation, this implies that the Pauli group $\pa$ decays into distinct orbits when $\cli$ acts on it by conjugation. This is in stark contrast to the case where $d$ is an odd prime, which is studied in Ref.~\cite{clark}, where there is only one non-trivial orbit. The orbit containing the elements $X_1$, $Z_1$, $\ldots$, $X_n$, $Z_n$ consists of elements of the form
\begin{align}\label{eq:form}
 \omg^{k+p/2} Z_1^{a_1}X_1^{b_1}\ldots Z_n^{a_n}X_n^{b_n}\,,
\end{align}
where $k,a_1,b_1,\ldots,a_n,b_n\in\ZZ_4$, at least one of the exponents $a_1$, $b_1$, $\ldots$, $a_n$, $b_n$ is odd, and $p=\sum_{i=1}^na_ib_i$ determines whether integer or half-integer powers of $\omg$ appear as phases (we sometimes write $\omg$ for $i$ to avoid confusion with indices).

The following proof is an adaption of the proof in Appendix A of Ref.~\cite{clark}, where it is shown that a certain set of gates generate the Clifford group $\cli$ for the case where $d$ is an odd prime. 
The general structure of our proof is identical to the one in Ref.~\cite{clark}, while the generating set and individual lemmas and their proofs are different.
After completion of this work, we became aware of the more general proof in Ref.~\cite{farinholt}.

Let us define the single-qudit unitaries 
\begin{align}
 H = \frac{1}{2}\sum_{j,k=0}^3\omg^{jk}\ket{j}\bra{k} 
\end{align}
\begin{align}
 S=\sum_{j=0}^3\omg^{j^2/2}\ket{j}\bra{j}\,,
\end{align}
and
\begin{align}
 T&=\half e^{-i\pi/4}\sum_{j,k=0}^3e^{i\frac{\pi}{4}(j-k)^2}\ket{j}\bra{k}\nn\\&=\frac{1}{2}
\begin{pmatrix}
\sqrt{i} & 1 & -\sqrt{i} & 1 \\
1 & \sqrt{i} & 1 & -\sqrt{i} \\
-\sqrt{i} & 1 & \sqrt{i} & 1 \\
1 & -\sqrt{i} & 1 & \sqrt{i}
\end{pmatrix}\,.
\end{align}

\begin{lemma}\label{lem1}
 The gates $S\mdag$, $T\mdag$, $Z\mdag$, $X$, $X\mdag$, and $\sqrt{i}H$ can all be generated from $S$, $T$, and $Z$.
\end{lemma}

\begin{proof}
 As $S^8=T^8=\id$, we have $S\mdag=S^7$ and $T\mdag=T^7$. 
 We have $\sqrt{i}H=STS=TST$ and $X=H\mdag ZH=(\sqrt{i}H)\mdag Z(\sqrt{i}H)$.
 Finally, $X^4=Z^4=\id$, so $X\mdag=X^3$ and $Z\mdag=Z^3$.
\end{proof}

If for some Clifford gate $U\in\cli$ we have
\begin{align}
 U(Z_1^{a_1}X_1^{b_1}\ldots Z_n^{a_n}X_n^{b_n})U\mdag = \alpha Z_1^{a_1'}X_1^{b_1'}\ldots Z_n^{a_n'}X_n^{b_n'}\,, 
\end{align}
with $|\alpha|=1$, we write 
\begin{align}
 M(U)(a_1,b_1,\ldots,a_n,b_n)^T=(a_1',b_1',\ldots,a_n',b_n')^T\,.
\end{align}
The matrices $M(U)\in\ZZ_d^{2n\times2n}$ form a representation of $\cli$, as $M(UV)=M(U)M(V)$.

\begin{lemma}\label{lem2}
 The gates $S$, $T$, and $Z$ generate the entire single-qudit Clifford group $\CC^{\otimes 1}_4$.
\end{lemma}

\begin{proof}
For some $U\in\CC^{\otimes 1}_4$, let $M(U)=\bigl(\begin{smallmatrix}a&c\\b&d\end{smallmatrix}\bigr)$.
Preserving the commutation relations of the single-qudit Pauli operators requires that $ad-bc=1\text{ (mod 4)}$.
One verifies that there are only $48$ matrices $M$ in the matrix ring $\ZZ_4^{2\times2}$ satisfying the requirement $\det M=1 \text{ (mod 4)}$.
We have $SXS\mdag=\sqrt{\omg}XZ$ and $TZT\mdag=\sqrt{\omg}ZX\mdag$, such that $M(S)=\bigl(\begin{smallmatrix}0&1\\1&1\end{smallmatrix}\bigr)$ and $M(T)=\bigl(\begin{smallmatrix}-1&1\\1&0\end{smallmatrix}\bigr)$.
Once can verify by brute force that products of at most $9$ factors $M(S)$ and $M(T)$ generate all of the aforementioned $48$ matrices.
Finally, since $XZX\mdag=\bar{\omg}Z$ and $ZXZ\mdag=\omg X$, we can generate arbitrary phases compatible with Eq.~(\ref{eq:form}) (for $n=1$).
\end{proof}

We define the controlled Pauli-operators 
\begin{align}
 C_X = \sum_{j=0}^3\ket{j}\bra{j}\otimes X^j
\end{align}
and
\begin{align}
 C_Z = \sum_{j=0}^3\ket{j}\bra{j}\otimes Z^j = \sum_{j,k=0}^3\omg^{jk}\ket{j}\bra{j}\otimes\ket{k}\bra{k}\,. 
\end{align}
Note that $C_Z$ has been called $\Lambda$ in the main part of this work.
We write $A\mapsto_UB$ as a shorthand for $UAU\mdag=B$.

We have
\begin{align}
 &Z_1\mapsto_{C_X}Z_1 \nn\\ &X_1\mapsto_{C_X}X_1X_2 \nn\\ &Z_2\mapsto_{C_X}Z_1\mdag Z_2 \nn\\ &X_2\mapsto_{C_X}X_2
\end{align}
and 
\begin{align}
 &Z_1\mapsto_{C_Z}Z_1 \nn\\ &X_1\mapsto_{C_Z}X_1Z_2 \nn\\ &Z_2\mapsto_{C_Z}Z_2 \nn\\ &X_2\mapsto_{C_Z}Z_1X_2\,,
\end{align}
showing that $C_X, C_Z\in\CC^{\otimes2}_4$.

\begin{lemma}\label{lem3}
 The gate $C_X$ can be generated from $S$, $T$, and $C_Z$.
\end{lemma}

\begin{proof}
We note that
\begin{align}\label{eq:hadamard}
 HXH\mdag=Z \qquad\text{and}\qquad HZH\mdag=X\mdag\,.
\end{align}
Thus
 \begin{align}
  C_X = H_2\mdag C_Z H_2 = (\sqrt{i}H_2)\mdag C_Z (\sqrt{i}H_2)\,,
 \end{align}
which together with Lemma~\ref{lem1} completes the proof.
\end{proof}

Let us define a more general controlled operator as
\begin{align}
 C_{st} = S_1^{-st}(C_X)^s(C_Z)^t\,.
\end{align}
It acts by conjugation as
\begin{align}\label{eq:control}
 &Z_1\mapsto_{C_{st}}Z_1 \nn\\ &X_1\mapsto_{C_{st}}\omg^{st/2}X_1X_2^sZ_2^t \nn\\ &Z_2\mapsto_{C_{st}}Z_1^{-s}Z_2 \nn\\ &X_2\mapsto_{C_{st}}Z_1^tX_2\,.
\end{align}
Up to the phase $\omg^{st/2}$, this action is identical to the one of the conditional Pauli gate $C_{X^sZ^t}$ studied in Ref.~\cite{clark}. We point out again that such a phase is unavoidable for $\ZZ_4$, as there is, for instance, no unitary $U$ such that $X_1\mapsto_UX_1X_2Z_2$, since these two operators are not isospectral.

Let us define the SWAP gate $\mathcal{S}$ via $\mathcal{S}\ket{j}\ket{k}=\ket{k}\ket{j}$. Evidently, it acts as
\begin{align}
 X_1\mapsto_\mathcal{S}X_2\,,\quad Z_1\mapsto_\mathcal{S}Z_2\,,\quad X_2\mapsto_\mathcal{S}X_1\,,\quad Z_2\mapsto_\mathcal{S}Z_1\,.
\end{align}
The gate $\mathcal{S}$ thus allows to generate non-local entangling gates from nearest-neighbor ones.

\begin{lemma}\label{lem4}
 The gate $i\mathcal{S}$ can be generated from $S$, $T$, and $C_Z$.
\end{lemma}

Note that the gates $\mathcal{S}$ and $i\mathcal{S}$ act identically by conjugation.

\begin{proof}
Let 
\begin{align}
 C_{X(1,2)}=\sum_{j=0}^3\ket{j}\bra{j}\otimes X^j\,,\quad C_{X(2,1)}=\sum_{j=0}^3X^j\otimes\ket{j}\bra{j}\,.
\end{align}
One verifies that 
\begin{align}
 C_{X(1,2)}C_{X(2,1)}\mdag C_{X(1,2)}(\sqrt{i}H_2)^2 = i\mathcal{S}\,,
\end{align}
which together with Lemmas \ref{lem1} and \ref{lem3} completes the proof.
\end{proof}

Let
\begin{align}\label{eq:PQ}
  P&=\alpha_PZ_1^{a_1}X_1^{b_1}\ldots Z_n^{a_n}X_n^{b_n} \nn\\ Q&=\alpha_QZ_1^{c_1}X_1^{d_1}\ldots Z_n^{c_n}X_n^{d_n}\,.
\end{align}
All arithmetics involving the exponents $a_j$, $b_j$, $c_j$, and $d_j$ that follow are to be understood modulo $4$.
It follows from the commutation relation $ZX=\omg XZ$ that $PQ = \omg^{(P,Q)}QP$ where
\begin{align}
 (P,Q) = \sum_{i=1}^na_id_i-b_ic_i\,.
\end{align}

\begin{lemma}\label{lem5}
 Given $P,Q\in\mathcal{P}_4^{\otimes n}$ with $(P,Q)=1$, we can generate $U\in\mathcal{C}_4^{\otimes n}$ from $S$, $T$, and nearest-neighbor $C_Z$ such that
\begin{align}
 P&\mapsto_U \alpha_PZ^{a_1'}X^{b_1'}\ldots Z^{a_n'}X^{b_b'} \nn\\  Q&\mapsto_U \alpha_QZ^{c_1'}X^{d_1'}\ldots Z^{c_n'}X^{d_n'}\,,
\end{align}
with $|\alpha_P|=|\alpha_Q|=1$,
and there exists $j\in\lbrace1,\ldots,n\rbrace$ such that $a_j'd_j'-b_j'c_j'=1$.
\end{lemma}

\begin{proof}
Let $P$ and $Q$ be as in Eq.~(\ref{eq:PQ}). Since \begin{align}\sum_{i=1}^na_id_i-b_ic_i=1\end{align} by assumption, there exists $j$ such that \begin{align}r_j=a_jd_j-b_jc_j\in\lbrace+1,-1\rbrace\,.\end{align}
If $r_j=1$, we are done.
If $r_j=-1$, then there is $k\neq j$ with $r_k=2$ or $r_k=-1$.
From Lemma~\ref{lem2}, we know that from $S_i$ and $T_i$ we can generate single-qudit unitaries that change $\bigl(\begin{smallmatrix}a_i&c_i\\b_i&d_i\end{smallmatrix}\bigr)$ in arbitrary ways as long as $r_i=a_id_i-b_ic_i$ is preserved.
So up to gates that can be generated from $S_j$ and $T_j$, we can assume that $a_j=0$, $b_j=c_j=1$, and $d_j=0$.
If $r_k=2$ then, up to gates that can be generated from $S_k$ and $T_k$, we can assume that $a_k=1$, $b_k=c_k=0$, and $d_k=2$.
Finally, if $r_k=-1$ then, up to gates that can be generated from $S_k$ and $T_k$, we can assume that $a_k=1$, $b_k=1$, $c_k=-1$, and $d_k=2$.
We note that application of a phase-gate $C_{Z(j,k)}$ changes $r_j$ to $r_j'=r_j+(b_kd_j-b_jd_k)$, and recall that non-local phase gates $C_{Z(j,k)}$ can be generated from nearest-neighbor ones and SWAP gates, which we can generate according to Lemma~\ref{lem4}.
In both cases ($r_k=2$ and $r_k=-1$), we find that application of a phase-gate $C_{Z(j,k)}$ gives $r_j'=1$.
\end{proof}

\begin{lemma}\label{lem6}
 Given $P,Q\in\mathcal{P}_4^{\otimes n}$ with $(P,Q)=1$, we can generate $U\in\mathcal{C}_4^{\otimes n}$ from $S$, $T$, and nearest-neighbor $C_Z$ such that
\begin{align}
 P\mapsto_UZ\otimes P'\qquad\text{and}\qquad Q\mapsto_UX\otimes Q'\,,
\end{align}
with $P',Q'\in\mathcal{P}_4^{\otimes n-1}$.
\end{lemma}

\begin{proof}
Let $P$ and $Q$ be as in Eq.~(\ref{eq:PQ}). By Lemma~\ref{lem5}, we can assume that there is $j$ with $a_jd_j-b_jc_j=1$.
Employing Lemma~\ref{lem4}, we can perform a SWAP between qudits $1$ and $j$. Finally, we perform a single-qudit unitary $L$ on qudit $1$ which is such that $M(L)=\bigl(\begin{smallmatrix}d_j&-c_j\\-b_j&a_j\end{smallmatrix}\bigr)$.
As $\det M(L)=1\text{ (mod 4)}$, such a unitary $L$ can be constructed from $S$ and $T$ according to Lemma~\ref{lem2}.
We note that
\begin{align}
 Z^{a_j}X^{b_j}\mapsto_L\alpha_ZZ \qquad\text{and}\qquad Z^{c_j}X^{d_j}\mapsto_L\alpha_XX\,,
\end{align}
with $|\alpha_X|=|\alpha_Z|=1$,
which completes the proof.
\end{proof}

\begin{lemma}\label{lem7}
 For any $V\in\mathcal{C}^{\otimes n}_4$ we can construct $U$ from $S$, $T$, $Z$, and nearest-neighbor $C_Z$ such that $UX_1U\mdag=VX_1V\mdag$ and $UZ_1U\mdag=VZ_1V\mdag$.
\end{lemma}

\begin{proof}
Clearly,
\begin{align}
 (VZ_1V\mdag,VX_1V\mdag) = (Z_1,X_1) = 1\,,  
\end{align}
so by Lemma~\ref{lem6}, we can assume that $VX_1V\mdag=X\otimes P'$ and $VZ_1V\mdag=Z\otimes Q'$, up to gates that can be constructed from $S$, $T$, and nearest-neighbor $C_Z$.

Now let
\begin{align}\label{eq:PprimeQprime}
  P'&=\alpha_PZ_2^{a_2}X_2^{b_2}\ldots Z_n^{a_n}X_n^{b_n} \nn\\ Q'&=\alpha_QZ_2^{c_2}X_2^{d_2}\ldots Z_n^{c_n}X_n^{d_n}\,.
\end{align}
We define
\begin{align}
 C_{st,i} = S_1^{-st}(C_{X(1,i)})^s(C_{Z(1,i)})^t\,,
\end{align}
with $i\in\lbrace2,\ldots,n\rbrace$.
The gate
\begin{align}
 U = (\sqrt{i}H_1)U_Q(\sqrt{i}H_1)\mdag U_P\,,
\end{align}
with 
\begin{align}
 U_P = \prod_{i=2}^nC_{b_na_n,i} \quad\text{and}\quad U_Q = \prod_{i=2}^nC_{d_nc_n,i}\,,
\end{align}
can be constructed from $S$, $T$, and nearest-neighbor $C_Z$ according to Lemmas \ref{lem1}, \ref{lem3} and \ref{lem4}.

Using Eqs.~(\ref{eq:hadamard}) and (\ref{eq:control}), we find the sequences of mappings
\begin{align}
 X_1 &\mapsto_{U_P} X\otimes P' \mapsto_{(\sqrt{i}H_1)\mdag} Z\mdag\otimes P' \nn\\&\mapsto_{U_Q} Z^{-1-\sum_{i=2}^n(a_id_i-b_ic_i)}\otimes P' \nn\\&\mapsto_{\sqrt{i}H_1} X^{1+\sum_{i=2}^n(a_id_i-b_ic_i)}\otimes P'\,,
\end{align}
and
\begin{align}
 Z_1 &\mapsto_{U_P} Z_1 \mapsto_{(\sqrt{i}H_1)\mdag} X_1 \mapsto_{U_Q} X\otimes Q' \mapsto_{\sqrt{i}H_1} Z\otimes Q'\,,
\end{align}
up to phases.
Using again that $XZX\mdag=\bar{\omg}Z$ and $ZXZ\mdag=\omg X$, and that $X$ can ge generated according to Lemma~\ref{lem1}, allows us to generate arbitrary phases compatible with Eq.~(\ref{eq:form}).
Since
\begin{align}
 1 &= (Z_1,X_1) = (VZ_1V\mdag,VX_1V\mdag) = (Z\otimes Q',X\otimes P') \nn\\&= 1+\sum_{i=2}^n(a_id_i-b_ic_i)\,,
\end{align}
we finally conclude that 
\begin{align}
 X_1 &\mapsto_U X\otimes P'=VX_1V\mdag \nn\\ Z_1 &\mapsto_U Z\otimes Q' = VZ_1V\mdag\,,
\end{align}
as required.
\end{proof}

\begin{theorem}
 Any Clifford gate $V\in\mathcal{C}^{\otimes n}_4$ can be constructed from $S$, $T$, $Z$, and nearest-neighbor $C_Z$.
\end{theorem}

\begin{proof}
The proof is done by induction over $n$. The case $n=1$ is given by Lemma~\ref{lem2}. For $n>1$, let $U$ be as in Lemma~\ref{lem7}. Since $U\mdag V$ commutes with $X_1$ and $Z_1$, we have $U\mdag V=\id\otimes V'$, where $V'\in\mathcal{C}^{\otimes n-1}_4$ acts on qudits $\lbrace2,\ldots,n\rbrace$. Assuming that the induction hypothesis holds for $n-1$, $V'$ and hence $V$ can be constructed from $S$, $T$, and nearest-neighbor $C_Z$.
\end{proof}

\section{Defect lines and holes}\label{app:holes}

Quantum computation in surface codes often uses the concept of `hole' defects \cite{raussendorf,fowler_hole,wootton_hole2}. These are extended areas in which a single anyon can reside. Their large size makes it difficult to measure their anyon occupancy, and also to change it without leaving a trace nearby. This allows them to store an additional logical qubit in a topologically protected manner. The code distance is given by the size of the hole (for $Z$ errors) and the distance to the nearest to its neighbour (for $X$ errors), and so can be made arbitrarily large. They are primarily considered in systems without a background Hamiltonian, where they are created and moved using measurements \cite{raussendorf,fowler_hole,wootton_hole2}. However, they can also be created by adiabatic means when a Hamiltonian is present \cite{wootton_hole1,cesare,zheng}. We now discuss this in detail for our system.

\subsection{Enlarging and shrinking holes}

The stabilizer is generated by the plaquette operators $M_p$ and $E_p$ for all triangular and hexagonal plaquettes. Let us consider, however, removing some of these operators from the stabilizer. In terms of the Hamiltonian, this means removing their corresponding terms.

Specifically, let us remove the plaquette operators $E_p$ and $E_q$ for two triangular plaquettes $p$ and $q$. This will open up a new fourfold degeneracy in the stabilizer space. Corresponding $Z$ basis states $\ket{g}$ can be labelled by the $\omega^g$ eigenstates of $E_p$. These states are therefore distinguished by the type of $e_g$ anyon residing in the hole. Due to conservation of anyons, the antiparticle $e_{-g}$ must reside in $q$. A further fourfold degeneracy will arise from each additional triangular plaquette removed from the plaquette. This is because only one removed plaquette, such as $q$, needs to have its occupation determined by the conservation of anyons.

For a logical qudit encoded in the additional stabilizer space, an $X$ type operation corresponds to creating a particle/antiparticle pair of $e$ anyons. One is placed on $p$ and the other on $q$. The number of sites on which this process has support, and so the number of sites on which noise must act in order to cause a logical $X$ error, is the distance between the two plaquettes. The qudit will therefore be topologically protected against such errors as long as the plaquettes are well separated.

The logical $Z$ of the logical qubit corresponds exactly to the operator $E_p$, and $E_{q}$ corresponds to $Z^\dagger$. Since these are three-body operators, this type of logical error requires action on only three qubit pairs. The stored qudit is therefore clearly not topologically protected against $Z$ type errors.

To address this problem, we can deform the lattice by making the plaquettes $p$ and $q$ larger. By making them arbitrarily large, logical $Z$ errors can be arbitrarily suppressed. 

To enlarge $p$, consider a neighbouring triangular plaquette $p'$. Let us use $j$ to denote the single site shared by these. If the stored qudit holds an arbitrary state $\ket{g}$, the state of the code will be a $\omega^g$ eigenstate of $E_p$. It will also be an $\omega^g$ eigenstate of $E_p E_{p'}$, since the state is a $+1$ eigenstate of the stabilizer $E_{p'}$.

Consider the adiabatic introduction of the term $X_j+X_j^\dagger$ to the Hamiltonian. This term should be much stronger than the adjacent plaquette operators, and so will effectively force $j$ into an eigenstate of $X$ and remove it from the code. Since this term does not commute with $E_p$ or $E_{p'}$, the resulting state will not be an eigenstate of these operators. However, the term does commute with the product $E_p E_{p'}$ since the support of the two plaquette operators on $j$ cancels in this product. It therefore remains the same $\omega^g$ eigenstate as it was for the initial state. The qudit state $\ket{g}$ has therefore been effectively transferred from the single plaquette $p$ to the combined plaquette $pp'$. The operator $E_p E_{p'}$ becomes the logical $Z$, and has support on four qubit pairs rather than three. This is further extended as more plaquettes are added using more $X_j + X_j^\dagger$ terms. As long as this is done for both $p$ and $q$, the qudit will become topologically protected against $Z$ errors as well as $X$.

The process used to extend holes can be reversed in order to shrink them. Consider a set of triangular plaquettes $p$, $p'$, $p''$, $\ldots$ that have been combined into a single hole. The basis state $\ket{g}$ of the qubit stored in this hole is associated with the $\omega^g$ eigenstate of $E_p E_{p'} E_{p''} \ldots$ We wish to shrink this hole so that $p$ is no longer a part of it. The state will then have a $+1$ eigenvalue for $E_p$, and the qudit state $\ket{g}$ will be associated with the $\omega^g$ eigenstate of $E_{p'} E_{p''} \ldots$

To achieve this, recall that the combination of the plaquettes in the hole is is enforced by the strong $X_j + X_j^\dagger$ terms on their shared sites. To remove $p$ from the hole, the term $E_p + E_p^\dagger$ should be added to the Hamiltonian, and the $X_j + X_j^\dagger$ term incident on $p$ should be removed. Doing this adiabatically will result in a final state with the $E_p$ term in its ground state, which is its $+1$ eigenspace. Due to conservation of anyon charge, the $e_g$ anyon that was held in the larger hole must still be held in the smaller hole. The qudit state therefore remains $\ket{g}$.

As well as this being true for each basis state $\ket{g}$, we must also be sure that the process preserves coherent superpositions. Any process that causes decoherence in this basis will correspond to measurement (by the environment) in the $Z$ basis. Any unitary that introduces unwanted relative phases can be expressed as a sum of powers of $Z$. As such, these processes must have support on all sites on which $Z$ has support, which are all sites around the hole. Since the shrinking (and expansion) of holes does not have such support, it cannot cause any decoherence.

Corresponding processes can also be applied to hexagonal plaquettes. In that case, a logical qudit can be stored in the $m_g$ occupations of plaquettes. Such qudits can be topologically protected by using lines of $Z_j + Z_j^\dagger$ terms to combine neighbouring horizontal plaquettes. 

Using the processes of extending and shrinking holes, it is possible to move them. Gates can then be implemented through braiding. Braiding an $e$-type hole of triangular plaquettes in state $\ket{g}$ around an $m$-type hole of hexagonal ones in state $\ket{h}$ corresponds to braiding the $e_g$ anyon held by the former around the $m_h$ of the latter, yielding a phase $\omega^{gh}$. This is a qudit generalization of the controlled phase gate.

\subsection{Fusing holes into defect lines}

By considering the alternative stabilizer generators $S$ and $R$ discussed in Sec.~\ref{sec:defect}, holes can also be created which hold $\psi_g$ anyons. These are formed by similarly combining the double plaquettes. Indeed, these are exactly the defect lines considered in the bulk of this paper. These have the property that anyons crossing the defect line undergo an automorphism that preserves the structure of the underlying Abelian state: it maps $e$ anyons to their dual, the $m$ anyons, and \emph{vice versa}. This property is not shared by the $e$- and $m$-type holes. In these cases, the lines form a boundary along which one type of anyon can condense, but the other cannot cross. It is this difference that gives the $\psi$-holes additional properties, namely the localized parafermion modes at their endpoints, that the $e$- and $m$-holes do not possess. The topological degeneracy and protection, however, is a property shared by all three.

An $e$-type hole and an $m$-type hole together correspond to a two-qudit space. However, let us consider the subspace spanned by states $\ket{g,g}$. These are such that the $e$-type hole carries an $e_g$ anyon whenever the $m$-type hole carries an $m_g$. A single qudit can be stored in this subspace

Since $\psi_g = e_g \times m_g$, two holes as described above hold a net $\psi_g$. Similar fusion can also be applied to holes, as we will now show. Specifically an $e$-type and an $m$-type hole can be combined into a $\psi$-type hole, and a $\psi$-type hole can be split into an $e$-type and $m$-type one. 

These processes are in fact a simple generalization of the hole extension and shrinking processes described above. Suppose we have a defect line along double plaquettes $P$, $P'$, $P''$, $\ldots$ This stores a qudit whose basis states $\ket{g}$ are $\omega^g$ eigenstates of $W_P W_{P'} W_{P''} \ldots$ Let us now extend this line. However, rather than adding another double plaquette, we instead add a triangular plaquette $p$. This can be done by adiabatically introducing the strong term $X_j + X_j^\dagger$ to the Hamiltonian on a site shared by $p$ and the triangular part of $P$. This term does not commute with either $E_p$ or $W_P$. However, it does commute with their product. The final state will then have the qudit basis states defined by the operator $E_p W_P W_{P'} W_{P''} \ldots$. Further such processes can be used to extend the $e$-type part of the defect line. Corresponding processes on the hexagonal plaquettes can be used to grow an $m$-type part of the defect line.
An illustration is given in Fig.~\ref{fig:split}.

\begin{figure}[h]
\centering
    \includegraphics[width=.90\columnwidth]{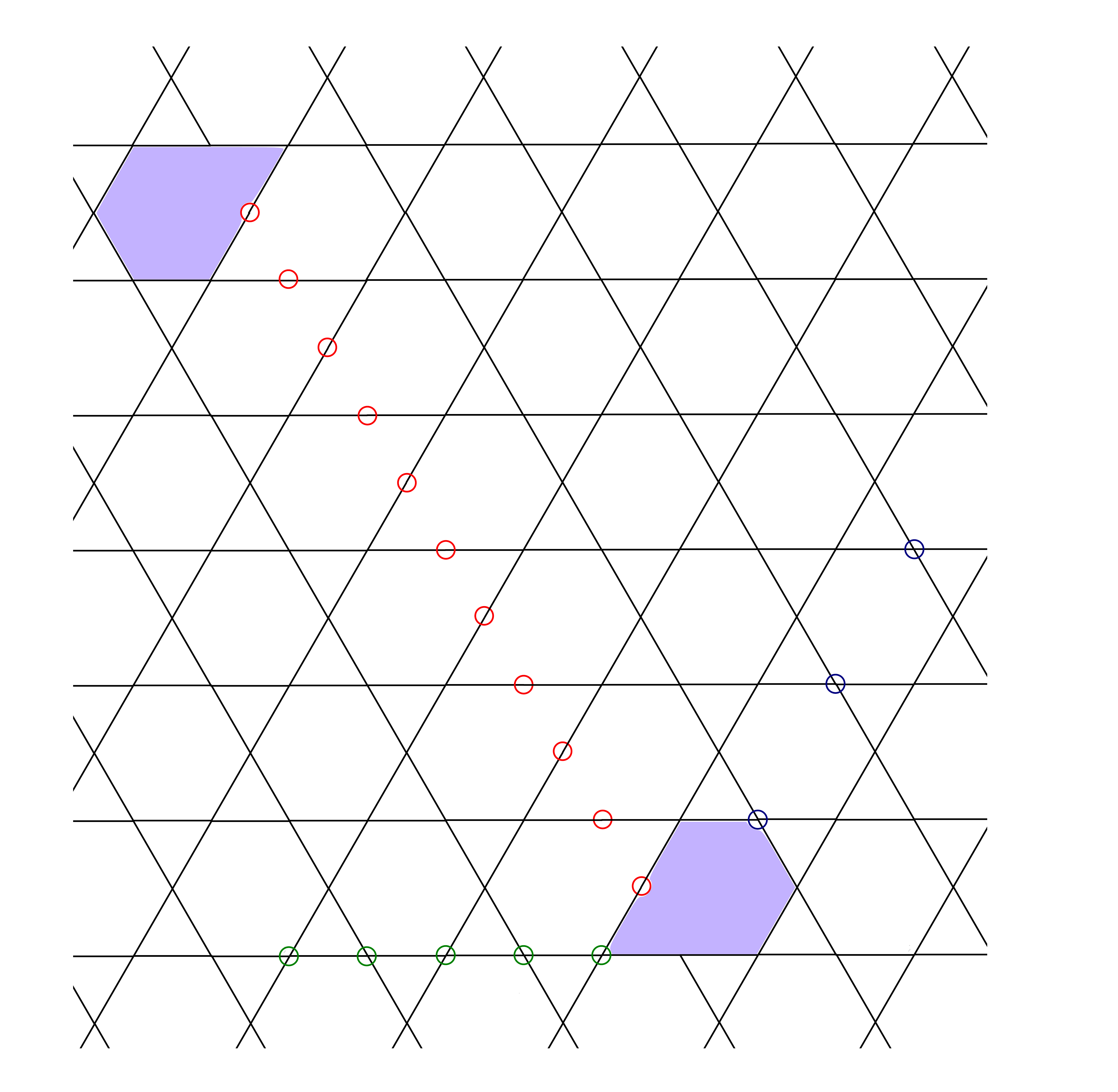} 
  \caption{Red circles denote defects as in Eq.~(\ref{eq:defects}), with a pair of parafermions (purple) emerging at the ends of the defect line. Blue circles correspond to terms of the form $X_j+X_j\mdag$, while green circles correspond to terms of the form $Z_j+Z_j\mdag$. These grow the $e$- and $m$-part of the defect line, respectively. Once both the blue and green defect lines have been added, the red defect line can be removed. The $\psi_g$ particle initially stored in the red defect line has then been split into its $e_g$ and $m_g$ components.}
  \label{fig:split}
\end{figure}

When both $e$-type and $m$-type parts have been added, the original $\psi$-type defect line can be removed. This is done simply by removing the parity operator terms along its length and allowing the $S_P$ terms to again dominate. The end result is that the $\psi_g$ originally stored in the defect line now resides in the $e$-type hole as an $e_g$ and the $m$-type hole as an $m_g$, corresponding to the state $\ket{g,g}$ of their individual qudits. As for the shrinking of holes, this process does not have sufficient support to distinguish between different basis states. The process therefore does not decohere any superpositions of these states, nor does it assign any relative phases.

To recombine the two holes into a single defect line, the process is simply reversed. This will be straightforward if the two holes are in a state of the form $\ket{g,g}$, since the state of the defect line will simply become $\ket{g}$. However some processes, such as mistakes during error correction, could result in holes whose states are not of this form. This will introduce frustration that will not allow all of the involved triangular and hexagonal plaquettes to return to their ground state after recombination.

As an example, consider a state of the form $\ket{g,h}$. This corresponds to an $e_g$ and an $m_{h}$, and could arise from an initial state $\ket{g,g}$ if an $m_{h-g}$ were added in error to the $m$-type hole, or from $\ket{h,h}$ with an $e_{g-h}$ error on the $e$-type hole.

The anyons $e_g$ and an $m_{h}$ can combine either to a $\psi_g$ and $m_{h-g}$, or a $\psi_{h}$ and $e_{g-h}$. The former will be energetically favourable, due to the weaker strength of the $M_p$ terms. The adiabatic process will therefore result in $\psi_g$ being stored on the defect line and an $m_{h-g}$ anyon present as an excitation on one of the triangular plaquettes that was once part of the holes. Syndrome measurement will then detect the $m_{g-h}$. However, since the value $g-h$ gives information only about the error that occurred, and not the value of $g$ or $h$, this does not extract any information about the stored qudit.


\begin{thebibliography}{99}
\bibitem{comment}
This is indeed a route we have tentatively followed. Unfortunately, the Hamiltonians necessary to generate $D(\ZZ_3)$ quantum double models in a qubit system turned out to be much more involved than Eq.~(\ref{eq:ham}).
\bibitem{note}
Note that, for closed boundary conditions, this added degeneracy does not arise for the first pair to be unpaired.
\end{thebibliography}
\end{document}